\documentclass[11pt,a4paper]{article}
\usepackage{latexsym,amssymb,amsmath,amsthm,amsfonts,enumerate,verbatim,xspace,
exscale}

\input xy
\xyoption{all}
\CompileMatrices
\UseComputerModernTips

\usepackage{graphicx}
\usepackage{amsmath,amsbsy,amsfonts,amssymb}
 \usepackage{subfigure}
\usepackage{color}

\usepackage[margin=1cm,
font=small,format=plain,labelfont=bf,up,textfont=up]{caption}

\makeatletter
\def\blfootnote{\xdef\@thefnmark{}\@footnotetext}
\makeatother

\title{The motion of the 2D hydrodynamic Chaplygin sleigh \\ in the presence of  circulation}
\author{Yuri N. Fedorov$^a$, Luis C. Garc\'ia-Naranjo$^{b}$, Joris Vankerschaver$^{c, d}$}

\theoremstyle{plain}
\newtheorem{theorem}{Theorem}[section]

\newtheorem{proposition}[theorem]{Proposition}

\newtheorem*{theorem*}{Theorem}
\newtheorem{remarkth}[theorem]{Remark}
\theoremstyle{definition}

\newenvironment{example}[1][Example.]{\begin{trivlist}
\item[\hskip \labelsep {\bfseries #1}]}{\end{trivlist}}

\def\vv<#1>{\langle#1\rangle}

\def\R{\mathbb{R}}
\def\I{\mathbb{I}}
\def\Lag{\mathcal{L}}

\def\se{se}

 \newcommand{\ad}{\mbox{$\text{\upshape{ad}}$}}

\newcommand{\vecu}{{\bf u}}
\newcommand{\In}{\mathcal{I}}
\newcommand{\Jn}{\mathcal{K}}

\def\F{\mathcal{F}}
\def\M{\mathcal{M}}
\def\B{\mathcal{B}}

\begin{document}
\maketitle

\begin{abstract}
We consider the motion of a planar rigid body in a potential flow with circulation and 
subject to a certain nonholonomic constraint. This model is related to the design of underwater vehicles. 

The equations of motion admit a reduction to a 2-dimensional nonlinear
system, which is integrated explicitly. We show that the reduced system comprises both asymptotic and periodic 
dynamics separated by a critical value of the energy, and give a complete classification of types of the motion.
Then we describe the whole variety of the trajectories of the body on the plane.
\end{abstract}

\blfootnote{\noindent
$^a$ Department de Matem\'atica Aplicada I, Universitat Politecnica de Catalunya, Barcelona, E-08028 Spain; e-mail: Yuri.Fedorov@upc.edu\\
$^b$ Departamento de Matem\'aticas, ITAM, Rio Hondo 1, Mexico City 01000, Mexico;   e-mail: luis.garcianaranjo@gmail.com \\
$^c$ Department of Mathematics, University of California at San Diego, 9500 Gilman Drive, La Jolla CA 92093-0112, USA; e-mail: joris.vankerschaver@gmail.com\\
$^d$ Department of Mathematics, Ghent University, Krijgslaan 281, B-9000 Ghent, Belgium}

\tableofcontents

\section{Introduction and outline}

In this paper, we consider the motion of a planar rigid body surrounded by an irrotational perfect fluid.
It ia assumed that there is a given amount $\kappa$ of circulation around the body, and that the body is subject
to a certain nonholonomic constraint, which models a very effective keel or fin attached to the body.
In the absence of circulation this system was termed the {\em hydrodynamic Chaplygin sleigh} in  \cite{hydro-sleigh},
since in the absence of the fluid, the nonholonomic constraint models the effect of a sharp blade in the classical
Chaplygin sleigh problem \cite{Chaplygin} which prevents the sleigh from moving in the lateral direction.

The hydrodynamic Chaplygin sleigh in the presence of circulation was recently considered in \cite{GN-V},
where it was shown that it is an \emph{LL system} on a certain central extension of $SE(2)$ by $\R^3$,
where the cocycle encodes the effects of the circulation on the body.

Our model for the nonholonomic constraint, which respects the Lagrange-D'Alembert principle,
has been considered in the aerospace engineering community \cite{Rand}, while robotic models for underwater
vehicles taking into account the effects of circulation were considered in \cite{KeHu2006}.

\paragraph{History of the Kirchhoff equations.}

The motion of a rigid body in a potential fluid in the absence of external forces was first described
by Kirchhoff \cite{Ki1877}.  His crucial observation was that the effect of the fluid on the body could
be described entirely in terms of the added mass and added inertia terms,
which depend on the geometry of the body only, and can be calculated analytically for a wide class of body shapes.
Kirchhoff's solution was extended to the case of rigid bodies moving in potential flow with circulation by, among others,
Chaplygin \cite{Ch1933} and \cite{Lamb}, who derived the equations of motion for this system,
provided an explicit integration in terms of elliptic functions, and described qualitative features of the dynamics,
such as periodic motions.
In recent years, these ideas have been extended to the case of rigid bodies interacting with point vortices
\cite{cylvortices, Bor_Mam_Ram}, vortex rings \cite{ShShKeMa08} and other vortical structures,
and they have been used to describe underwater vehicles \cite{Leonard1997} and the motion of bio-organisms
\cite{ChMu2011, ChSi2008}.
A comprehensive overview of the history of these equations can be found in \cite{Borisov_circulation}.

\paragraph{Contributions of this paper.}

We  show that the hydrodynamic Chaplygin sleigh with circulation is a new example of a completely integrable nonholonomic system: we discuss  qualitative features of the  dynamics, and we explicitly integrate the reduced equations of motion.
\paragraph{Contents of the paper.}
In section~\ref{S:Preliminaries} we recall Kirchhoff's equations for a planar rigid body
moving in a potential fluid and consider the Chaplygin-Lamb equations for the motion of the body in
the presence of circulation.
For the purpose of completeness, the added masses are computed explicitly for a body of elliptical shape.

In section~\ref{S:Hydro-sleigh-circ}, the reduced equations of motion on
the coalgebra $\se(2)^*$ are written explicitly in terms
 of the added masses and the coefficient $\alpha$ that depends on the position and orientation of the
 fin (or blade) on the body. Again, for the purpose of completeness, we give explicit formulas
 for the parameters that enter the equations for a body having an elliptical shape.
 It is shown that
the presence of circulation around the body adds new features to the motion of the
hydrodynamic Chaplygin sleigh, and that there exists a critical value of the kinetic energy
of the fluid-body system
that divides periodic from asymptotic behavior. Indeed, for
 small, subcritical energy values, the reduced dynamics are
periodic. In this case  the circulation effects
drive the dynamics (via the so-called {\em Kutta-Zhukowski force}).
On the contrary, if the initial  energy  is supercritical, then
the inertia of the body overcomes the circulation effects and it evolves asymptotically
from one circular motion to another, where the limit circumferences have different radius.
This resembles  the motion of the body in the
absence of circulation treated in \cite{hydro-sleigh}. Moreover, we identify 7 regions
in the reduced phase space that yield distinct qualitative dynamics of the motion of the
sleigh on the plane.

In section~\ref{S:solution}  the reduced equations of motion for the hydrodynamic Chaplygin sleigh
with circulation are integrated explicitly for a generic sleigh. The  form of the solution varies according to
the energy regime. If the  energy is subcritical,
the solution is a quotient of trigonometric functions and we give an explicit expression for the period
in terms of the energy. We also give a closed expression for the angular part of the monodromy matrix
that is involved in the reconstruction process. This formula allows us to show that the
qualitative behavior of the sleigh on the plane is very sensitive to initial conditions for energy values
that are slightly subcritical.

On the other hand, for the critical value of the energy, the solution of the reduced equations is given as a rational
function of time. The closed expression for the solution is used to show that the motion of the
sleigh on the plane in this case is bounded and evolves from one circular motion to another, where the
limit circumferences have equal radius.  Finally, the solution of the reduced equations for supercritical energies is given as a quotient of hyperbolic functions. In this case we also
express the distance between the centers of the limit circumferences in terms of integrals
generalizing the Beta-function.

In Conclusions we motivate a further study of the hydrodynamic Chaplygin sleigh  in the
presence of point vortices.

\section{Preliminaries: Fluid-structure interactions}
\label{S:Preliminaries}

In this section, we give an overview of the Kirchhoff equations describing the dynamics of a rigid body in potential flow,
and of the Chaplygin-Lamb equations dealing with rigid bodies in the presence of circulation.  
Most of the material covered in this 
section can be found, for instance, in \cite{Kanso}
as well as in the classical works of Lamb \cite{Lamb} and Milne-Thomson \cite{MiTh1968}.

\subsection{Kinematics of rigid bodies and ideal fluids}
\label{S:Kirchhoff}

Following Euler's approach, consider an orthonormal \emph{body frame}
$\{{\bf E}_1, {\bf E}_2\}$ that is attached to the body. This frame is related to a fixed \emph{space frame}
$\{{\bf e}_1, {\bf e}_2\}$
by a rotation  by an angle $\theta$ that specifies the orientation of the two dimensional body at each time.
We will denote by ${\bf x}=(x, y) \in \R^2$ the spatial coordinates of the origin of the body frame (see Figure
\ref{F:kinematics-diagram}).
The configuration of the body at any time is
completely determined by the element  $g$ of the two dimensional Euclidean group
$SE(2)$ given by
\begin{equation*}
g= \left ( \begin{array}{ccc} \cos \theta  & - \sin \theta  & x \\ \sin \theta  & \cos \theta  & y \\ 0 & 0 & 1 \end{array} \right ) \in SE(2).
\end{equation*}

\begin{figure}[h]
\centering
\subfigure[Body frame is aligned with axes of symmetry of the body]{\includegraphics[width=5cm]{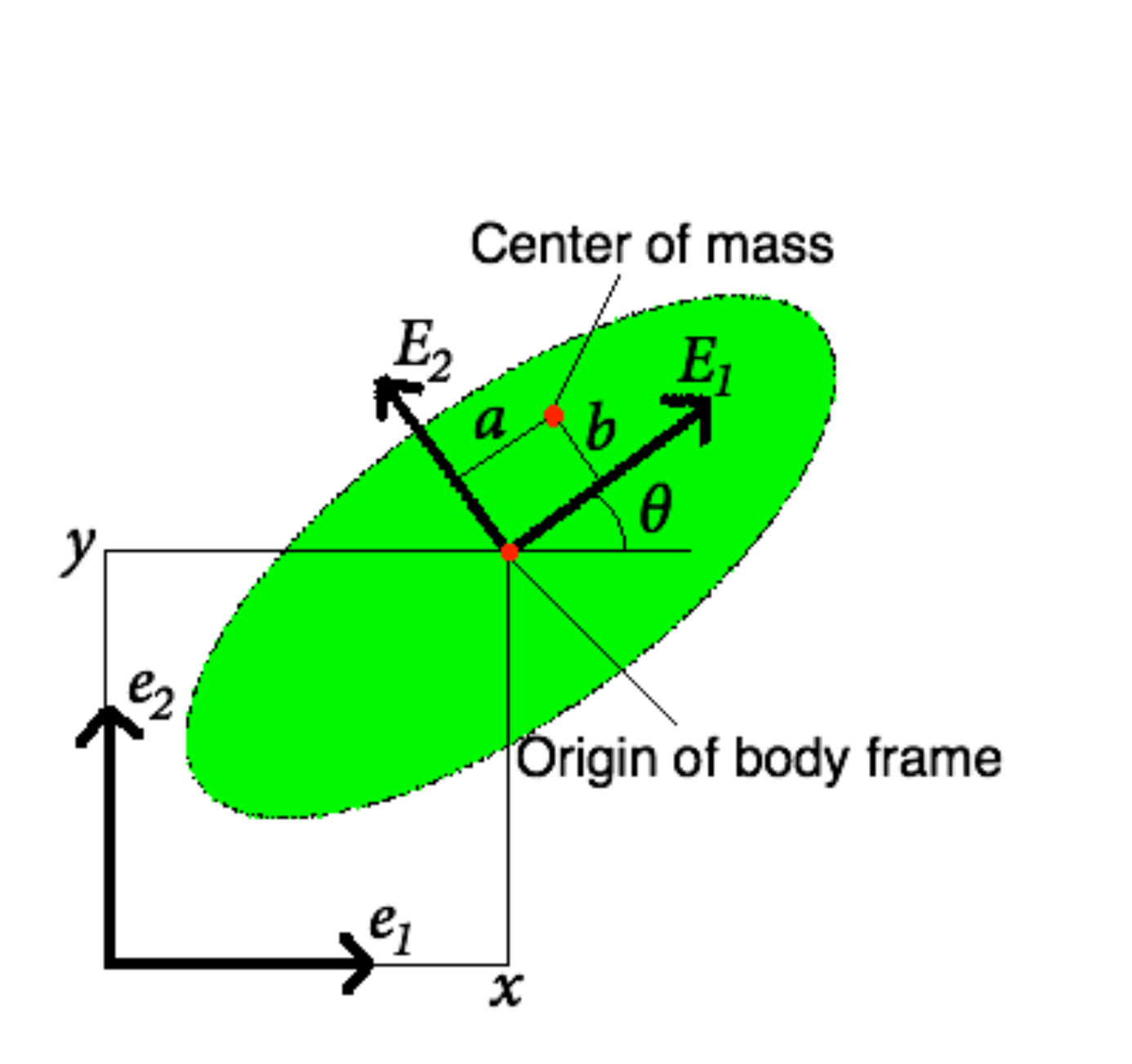}} \qquad \qquad
\subfigure[Arbitrary position and orientation of the body frame.]{\includegraphics[width=5cm]{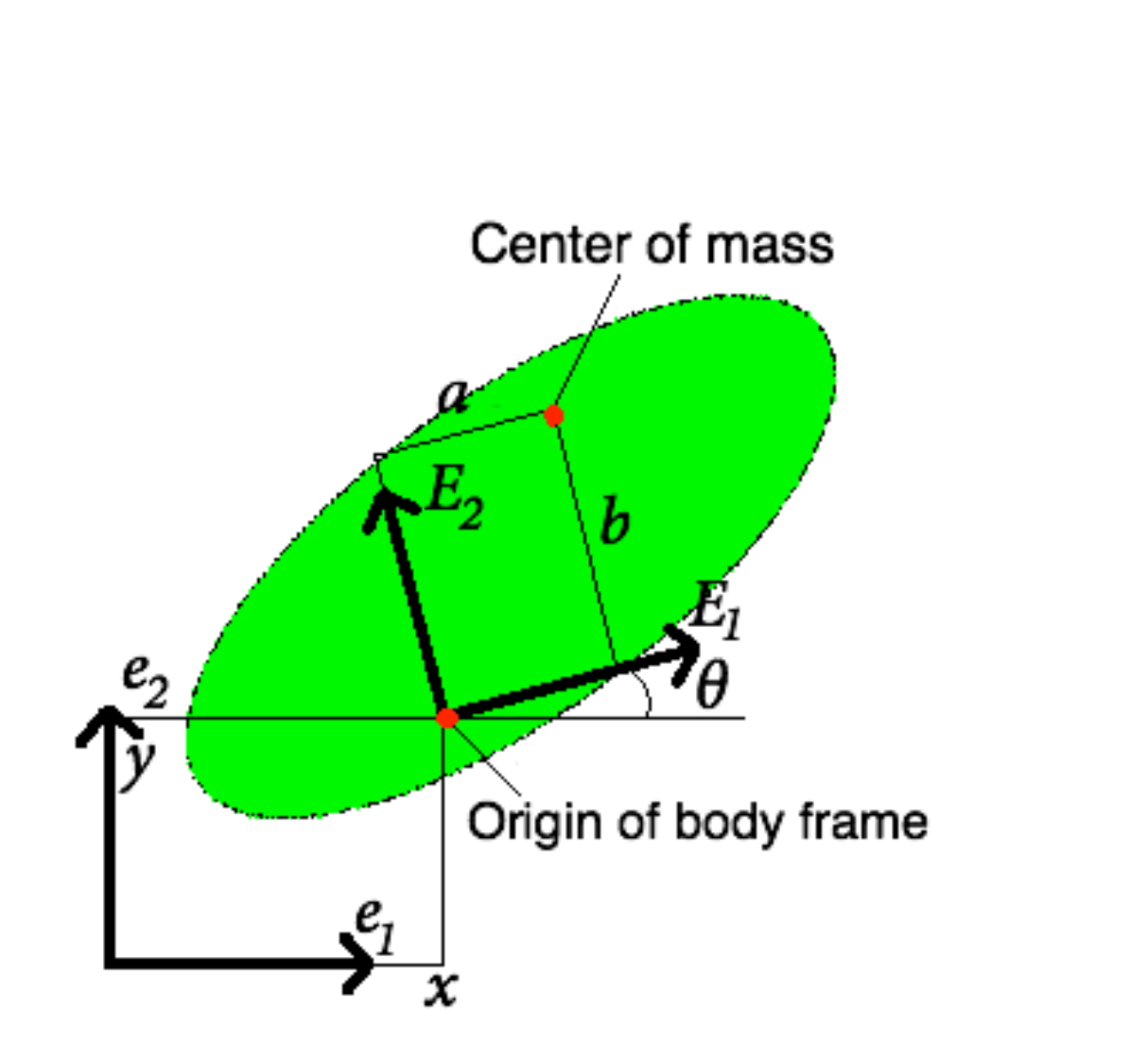}}
\caption{
\small{Two different choices of the body  frame for an elliptical  two-dimensional rigid body. In both cases the origin of
the body frame does not coincide with the center of mass. \label{F:kinematics-diagram}}
}
\end{figure}

We will often denote the  above element in $g\in SE(2)$ by $g=(R_{\theta}, {\bf x})$ where
$R_{\theta} \in \operatorname{SO}(2)$ is the  rotation matrix determined by the angle $\theta$.
 Let $(v_1, v_2)\in \R^2$ be the linear velocity of the origin of the body frame written in
the body coordinates, and denote by $\omega  = \dot \theta $ the body's angular velocity. They define the element $\xi $ in the Lie algebra $\se(2)$
given by
\begin{equation*}
  \xi = g^{-1}\dot g= \left ( \begin{array}{ccc} 0 & - \omega   & v_1 \\ \omega  & 0& v_2 \\ 0 & 0 & 0 \end{array} \right ) \in \se (2).
\end{equation*}
Explicitly we have
\begin{equation}
\label{E:Reconstruction_Equations_2D}
\dot \theta = \omega, \qquad v_1=\dot x \cos \theta  +  \dot y\sin \theta , \quad v_2= -\dot x\sin \theta   +  \dot y \cos \theta .
\end{equation}
For convenience, we will sometimes 
denote $\xi \in \se(2)$ as the column vector $(\omega, v_1, v_2)^T\in \R^3$. The Lie algebra commutator
takes the form
\begin{equation*}
	  \left [ \, (\omega, v_1,v_2) \, , \,  (\omega', v_1',v_2')\, \right ]_{\mathfrak{se}(2)} 	  = ( \, 0 \, , \,  v_2\omega'-\omega v_2' \, , \,  \omega v_1' -v_1 \omega ' \,  ).
\end{equation*}
The kinetic energy of the body is given by
\begin{equation}
\label{E:Body_Energy_in_terms_of_velocities}
T_\mathcal{B}= \frac{1}{2} \left ( (\In_\mathcal{B} +ma^2+mb^2) \omega^2 + m( v_1^2+v_2^2) -mb\omega v_1 +ma \omega v_2  \right ),
\end{equation}
where  $m$ is the mass of the body, $(a,b)$ are body coordinates of the center of mass
(see Figure \ref{F:kinematics-diagram}), and $\In_\mathcal{B} $
is the moment of inertia of the body about the center of mass. It is a positive
definite quadratic form on $\se (2)$ whose matrix  is the \emph{body inertia tensor}
 \begin{equation}
\label{E:Inertia_Matrices_2D}
\I_\mathcal{B}=\left ( \begin{array}{ccc}\In_\mathcal{B} + m(a^2+b^2) & -mb & ma \\ -mb   & m & 0 \\ ma& 0 & m \end{array} \right ).
\end{equation}

\paragraph{The fluid flow at a given instant.} Consider now the motion of the fluid that
surrounds the body. Suppose that at a given instant the body occupies a
region $\mathcal{B}\subset \R^2$.  The flow is assumed to take place in the connected unbounded
region  $\mathcal{U}:=\R^2\setminus \mathcal{B}$ that is not occupied by the body.

We  assume that  the flow is  potential so the  Eulerian velocity of the fluid $\vecu$ can be written as $\vecu =\nabla \Phi$ for
a  fluid potential $\Phi:\mathcal{U} \to \R$. Incompressibility of the fluid implies that
$\Phi$ is harmonic,
\begin{equation*}
\label{E:Laplace_Equation}
\nabla^2 \Phi =0 \qquad \mbox{on} \qquad \mathcal{U}.
\end{equation*}

The boundary conditions for $\Phi$ come from the following considerations. On the one hand it is assumed that, up to a purely circulatory flow around the
body,  the motion of the
fluid is solely due to the motion of the body. This assumption   requires the fluid velocity
$\nabla \Phi $ to vanish at infinity. Secondly, to avoid cavitation or penetration of the fluid into the body,
we require the normal component of the fluid velocity at a material point $p$ on the boundary of $\mathcal{B}$
to agree with the
normal component of the velocity of $p$. Suppose that the vector $( X,Y) \in \R^2$ gives   body coordinates for $p$. The
latter boundary condition is expressed as
\begin{equation*}
\label{E:Boundary_conditions}
\left .  \frac{\partial \Phi}{\partial n} \right |_{p\in \partial \mathcal{B}} =  (v_1-\omega Y)n_1
 +(v_2 +\omega X)n_2,
\end{equation*}
where  ${\bf n}=(n_1,n_2)$ is the outward unit  normal vector to $\mathcal{B}$ at  $p$ written
in body coordinates. These conditions determine the flow of the fluid up to a purely
circulatory flow around the body  that would persist if the body
is brought to rest. The latter is specified by the value of the circulation $\kappa$ around the body as we
now discuss.

The potential $\Phi$ that satisfies the above boundary value problem  can be
written in terms of the body's  velocities $v_1, v_2, \omega$, in \emph{Kirchhoff form}:
  \begin{equation}
\label{E:Kirch_Potential2D}
\Phi= v_1 \phi_1 +  v_2 \phi_2 +  \omega \chi +\phi_0,
\end{equation}
%
%

where $\phi_i$, $i=0,1,2$, and $\chi$ are harmonic functions on $\mathcal{U}$ whose gradients vanish at infinity and  satisfy:
\begin{equation*}
\label{E:Boundary_conditions_Kirch_form}
\left . \frac{\partial \phi_i}{\partial n} \right |_{\partial \mathcal{B}} =n_i, \; i=1,2, \qquad \left . \frac{\partial \chi} {\partial n} \right  |_{\partial \mathcal{B}} = Xn_2-Yn_1, \qquad \left . \frac{\partial \phi_0}{\partial n} \right |_{\partial \mathcal{B}} =0.
\end{equation*}
The potential $\phi_0$ is multi-valued and  defines the
circulatory flow around the body.
 The circulation $\kappa$ of the fluid around the body satisfies
\begin{equation*}
\label{E:defcirc}
\kappa = \oint_{\partial \mathcal{B}} \vecu \cdot d{\mathbf{l}}= \oint_{\partial \mathcal{B}} \nabla \phi_0 \cdot d{\mathbf{l}},
\end{equation*}
and remains constant during the motion. Figure \ref{Fig:streamlines} shows the streamline pattern of the flow determined by the motion of an elliptical
body  for different values of $\kappa$.
\begin{figure}[h]
\centering
\subfigure[$\kappa=0$]{\includegraphics[width=5cm]{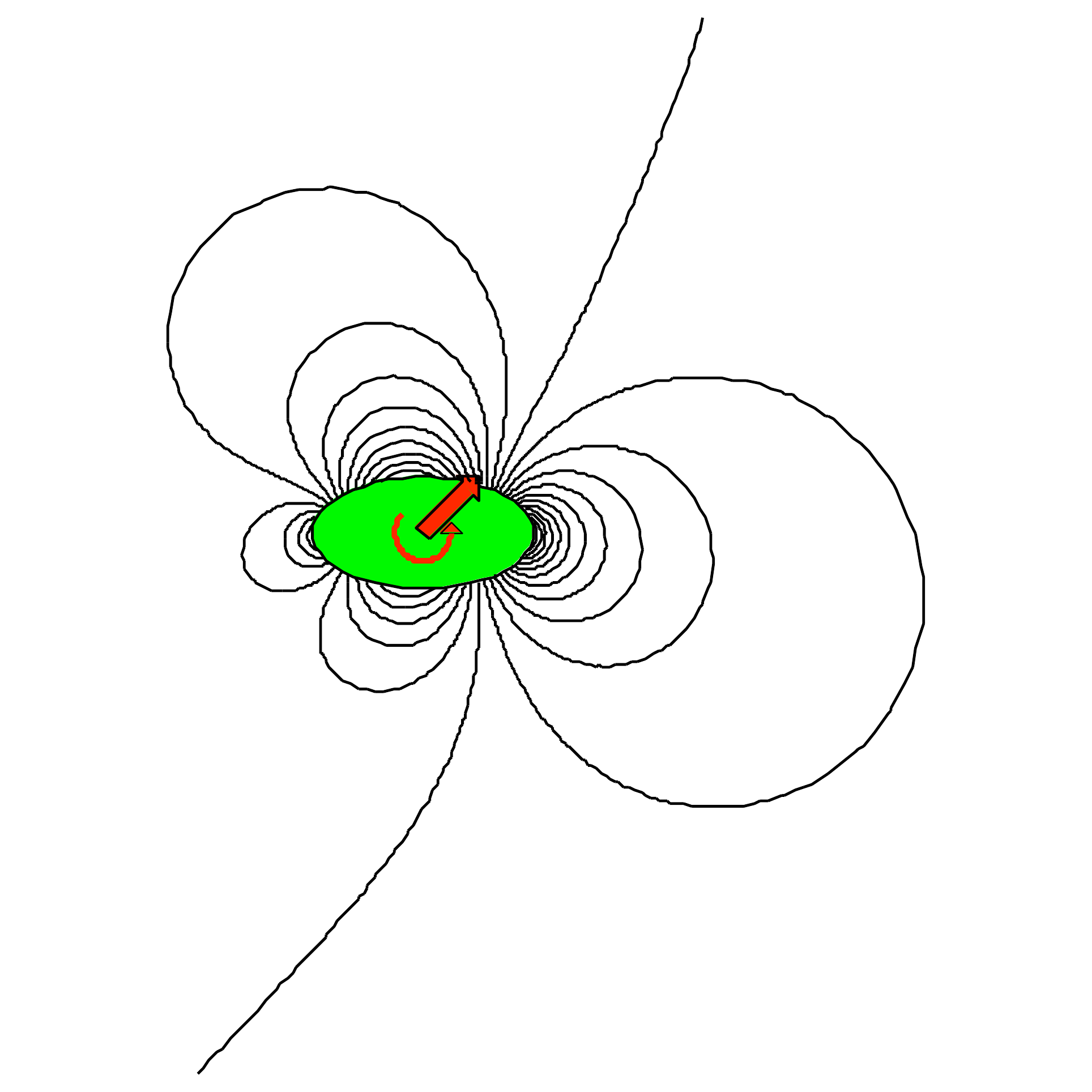}} \qquad
\subfigure[$\kappa=2\pi$]{\includegraphics[width=5cm]{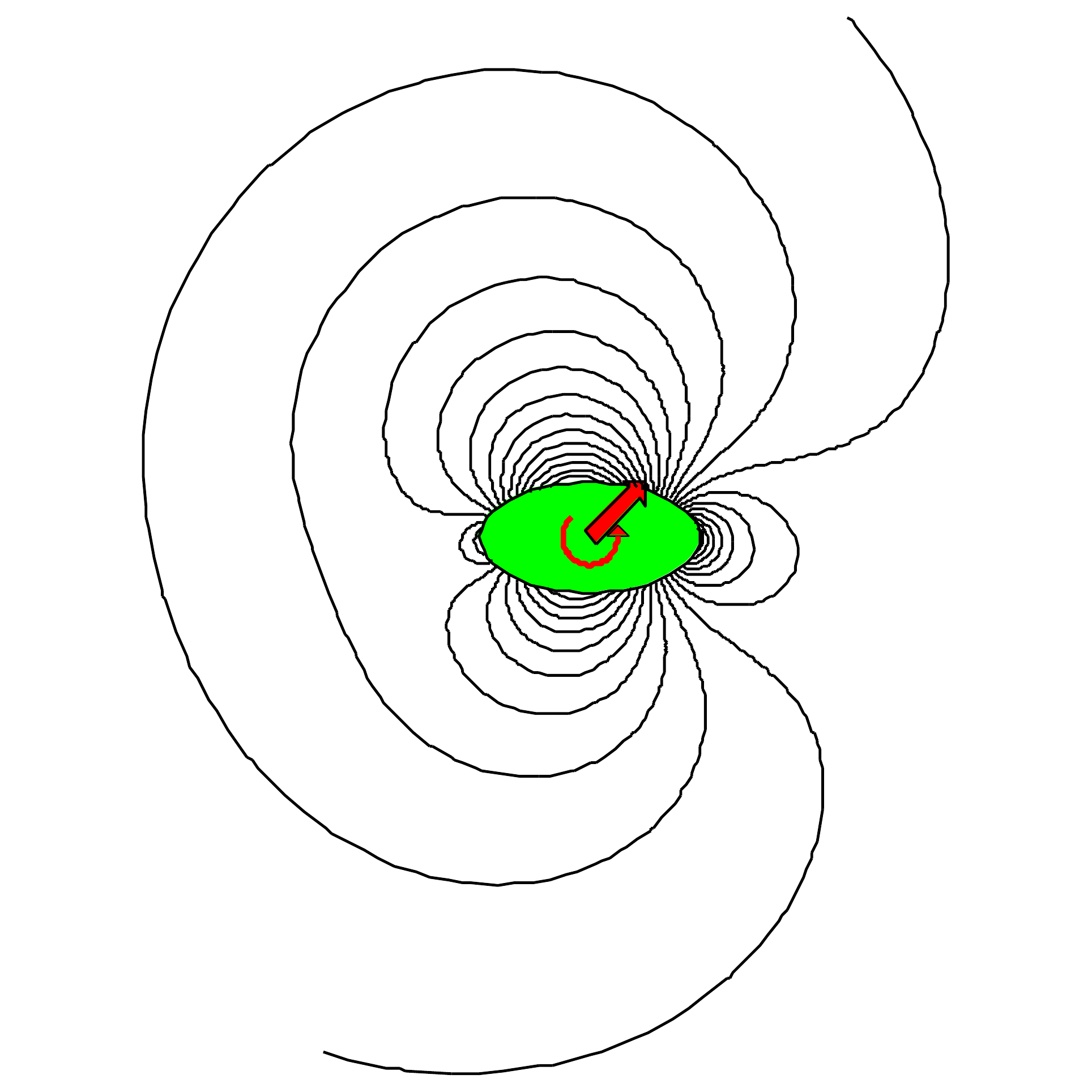}} \qquad
\subfigure[$\kappa=10$]{\includegraphics[width=5cm]{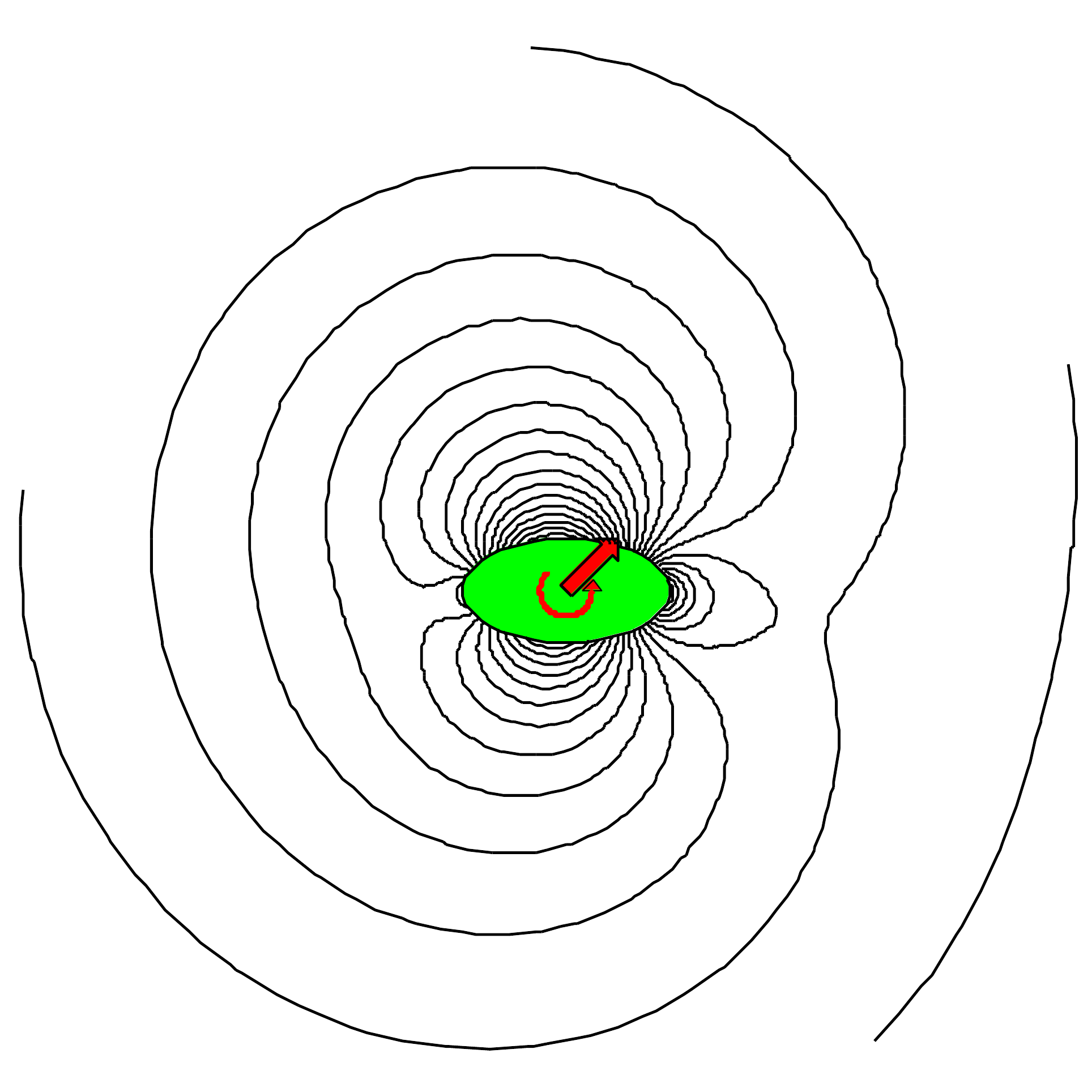}}
\caption{
\small{ Stream line pattern for an ellipse moving on the plane for different values of the circulation $\kappa$.  The major and minor semi-axes of the ellipse are  $A=2, B=1$. The body frame is aligned with the principal axes of the ellipse and the velocity of the body satisfies  $v_1=v_2=1$ and $\omega=5$. \label{Fig:streamlines}}
}
\end{figure}

Disregarding the  circulatory motion, the kinetic energy of the fluid  is given by
\begin{equation*}
T_\mathcal{F}=\frac{\rho}{2} \int_\mathcal{U}
 || \nabla (\Phi-\phi_0) ||^2 \, dA,
\end{equation*}
where $dA$ is the area element in $\R^2$ and  $\rho$ is the (constant) fluid density.  We have subtracted the circulatory part from the velocity potential, as it would give rise to an infinite contribution to the fluid kinetic energy \cite{Sa1992}.

By substituting  \eqref{E:Kirch_Potential2D} into the above, one can express $T_\mathcal{F}$
 as the quadratic form
\begin{equation}
\label{E:Fluid_Energy_in_terms_of_velocities}
T_\mathcal{F}= \frac{1}{2} \left ( \sum_{i,j=1}^2 \M^{ij}_\mathcal{F} v_{i}v_j + 2 \sum_{i=1}^2\Jn^{i}_\mathcal{F}v_i\omega +\In_\mathcal{F} \omega^2  \right ),
\end{equation}
where  $\mathcal{M}_\mathcal{F}^{ij}, \Jn_\mathcal{F}^{i}$, $i,j=1,2$, and $\In_\mathcal{F}$  are certain constants that only depend on the body shape.
Explicitly one has (see \cite{Lamb} for details),
\begin{equation*}
\begin{split}
& \M^{ij}_\mathcal{F}=-\rho \int_{\partial \mathcal{B}} \phi_i \frac{\partial \phi_j}{\partial n} \, dl=-\rho \int_{\partial \mathcal{B}} \phi_j \frac{\partial \phi_i}{\partial n} \, dl , \; i,j=1,2,  \qquad \In_\mathcal{F}=-\rho \int_{\partial \mathcal{B}} \chi \frac{\partial \chi}{\partial n} \, dl
 \\ & \Jn^i_\mathcal{F}=-\rho \int_{\partial \mathcal{B}} \phi_i \frac{\partial \chi}{\partial n} \, dl = -\rho \int_{\partial \mathcal{B}} \chi \frac{\partial \phi_i}{\partial n} \, dl, \; i=1,2.\end{split}
\end{equation*}
These constants are referred to as \emph{added masses} and are conveniently written
in $3\times 3$ matrix form to define the (symmetric) \emph{added inertia tensor}:
\begin{equation*}
\I_\mathcal{F}:=\left ( \begin{array}{cc} \In_\mathcal{F} & \Jn_\mathcal{F} \\ \Jn_\mathcal{F}^T  & \mathcal{M}_\mathcal{F} \end{array} \right ),
\end{equation*}
that defines $T_\mathcal{F}$ as a quadratic form on $\se(2)$.

\begin{example} For an elliptic rigid body with semi-axes of length $A>B>0$, the added masses and moments of inertia take on a particularly convenient form.  The kinetic energy of the fluid is given by (see \cite{Lamb})
\begin{equation*}
T_{\mathcal{F}}= \frac{\rho \pi }{2} \left ( B^2 v_1^2 +A^2v_2^2 + \frac{(A^2-B^2)^2}{4}\omega^2 \right ),
\end{equation*}
where we have ignored the circulatory motion around the body. The corresponding added inertia tensor is
thus given by
\begin{equation}
\label{E:Added_Inertia_Ellipse}
\I_\mathcal{F}=\rho \pi \left ( \begin{array}{ccc} \frac{(A^2-B^2)^2}{4} & 0 & 0 \\ 0 &   B^2 & 0 \\ 0 & 0 &  A^2\ \end{array} \right ).
\end{equation}
We emphasize that this particular form of the added inertia tensor was derived under the assumption that the body frame is aligned with the symmetry axes of the ellipse ($\nu=0$, $r=s=0$ in Figure~\ref{F:SleighDiagram} ahead).  When this is not the case, the added mass tensor is more complicated, and in particular need not be diagonal, as is shown in \eqref{E:Added_Inertia_Ellipse_Rot_Axes}.
\end{example}

\subsection{Rigid body motion in potential flow}

The total kinetic energy, $T$,  of the solid-fluid system (excluding the circulatory motion)
is the sum of the kinetic energy $T_\B$ of the rigid body and the
 energy $T_\F$ of the fluid.  As both $T_\B$ and $T_\F$ are functions on $T SE(2)$, so is the total energy $T$.  In the absence of external forces or circulation, the Lagrangian $\Lag$ of the solid-fluid system is just the kinetic energy: $\Lag = T$, and in this case, the motion of the rigid body describes a geodesic curve in
 $SE(2)$ with respect to the Riemannian metric defined by $\Lag$.

 In view of
 (\ref{E:Body_Energy_in_terms_of_velocities}) and (\ref{E:Fluid_Energy_in_terms_of_velocities}), we can write the Lagrangian $\Lag=T_\B+T_\F$
in terms of the linear and angular velocities of the body (written in the body frame) and this expression  does not  depend on the particular position and orientation of the body, i.e. is independent of the group element $g=(R_\theta, {\bf x})\in SE(2)$.
Thus $\Lag$ is  invariant under the lifted action  of left multiplication on $SE(2)$. This symmetry corresponds to invariance under translations and rotations of the space frame. The reduction of this symmetry defines Euler-Poincar\'e equations on the Lie algebra $\se(2)$ or, in the Hamiltonian setting, the Lie-Poisson equations on the coalgebra $\se(2)^*$. The latter are precisely Kirchhoff's equations that are explicitly written below.

The invariance of $\Lag$ allows us to define a function
$L$ on $\se(2)$, called the \emph{reduced Lagrangian}.  Explicitly, $L$ is given by
\begin{equation*}
\label{E:redlag}
 	L(\xi)= \frac{1}{2} \xi^T \I \xi,
\end{equation*}
where $\xi=(\omega, v_1, v_2)^T\in \se(2) \cong \R^3$ is  a  column vector and
the matrix $\I$ is the sum of the inertia matrix $\I_B$ of the rigid body and the added masses and inertia $\I_\mathcal{F}$ of the fluid: $\I =\I_\mathcal{B}+\I_\mathcal{F}$.

Since $\se(2)$ is isomorphic to $\mathbb{R}^3$ and using the Euclidian inner product, we
identify $\se(2)^\ast \cong \mathbb{R}^3$.  A typical element $\mu$ is represented as a row vector $\mu = (k, p_1, p_2)$.  The duality pairing between $\mu$ and an element $\xi= (\omega, v_1, v_2)^T$ of $\se(2)$ is given by
\begin{equation*}
\label{E:pairing_in_se(3)}
\langle \mu, \xi \rangle =\mu \xi=k\omega +p_1v_1+p_2v_2.
\end{equation*}

With this identification, the \emph{Legendre transform} associated to $L$ is defined as the mapping
\[
	\mathbb{F} L : \se(2) \rightarrow \se(2)^\ast
\]
given by $\mathbb{F} L (\xi) = \mu$, where $\mu = (\I \xi)^T$.  The components of $\mu = (k,p_1,p_2)$ are explicitly given by
\begin{equation}
\label{E:Explicit_Legendre}
\begin{split}
k&=(\mathcal{I}_\mathcal{B}+m(a^2+b^2) +\mathcal{I}_\mathcal{F})\omega + (-mb + \Jn_\mathcal{F}^1) v_1 +(ma + \Jn_\mathcal{F}^2) v_2, \\
p_1&=(-mb + \Jn_\mathcal{F}^1) \omega + (m+ \mathcal{M}^{11}_\mathcal{F})v_1 +  \mathcal{M}^{12}v_2 , \\
p_2&=(ma + \Jn_\mathcal{F}^2) \omega + \mathcal{M}^{12}v_1 +   (m+ \mathcal{M}^{22}_\mathcal{F})v_2.
\end{split}
\end{equation}
In classical hydrodynamics $k$ and $(p_1,p_2)$ are known as ``impulsive pair" and ``impulsive force" respectively.

The reduced Hamiltonian $H:\se(2)^*\rightarrow \R$ is given by
\begin{equation*}
\label{E:Hamiltonian}
H(\mu)=\frac{1}{2}\mu \I^{-1}\mu^T,
\end{equation*}
and the corresponding (minus) Lie-Poisson equations are
$
\dot \mu=\ad^*_{\I^{-1} \mu}\mu.
$
Written out in component form, these equations are nothing but the \emph{Kirchhoff equations}:
\begin{equation}
\label{E:Kirchhoff_2D}
\begin{split}
\dot k &= v_2p_1-v_1p_2, \\
\dot p_1 &= \omega p_2, \quad  \quad
\dot p_2 = - \omega p_1,
\end{split}
\end{equation}
where the velocities $(\omega, v_1, v_2)^T$ and the impulses $(k, p_1, p_2)$ are related by the Legendre transformation \eqref{E:Explicit_Legendre}.  Finally, we remark that the motion of the body in space can be found from a solution of \eqref{E:Kirchhoff_2D} by solving the \emph{reconstruction equations} \eqref{E:Reconstruction_Equations_2D}.

%

\subsection{Rigid body motion with circulation}
\label{sec:circ}

In the presence of circulation, the Kirchhoff equations on $\se(2)^*$ have to be modified to include the Kutta--Zhukowski force mentioned in the introduction.  This is a gyroscopic force, which is proportional to the circulation $\kappa$.  In this case, the equations of motion are referred to as the \emph{Chaplygin-Lamb equations} \cite{Ch1933, Lamb}, and they are given by
\begin{equation}
\label{E:Kirchhoff_1_2D_circulation}
\begin{split}
\dot k &= v_2p_1-v_1p_2 - \rho(\alpha v_1 + \beta v_2), \\
\dot p_1 &= \omega p_2 - \kappa \rho v_2 + \rho \alpha \omega, \\
\dot p_2 &= - \omega p_1 + \kappa \rho v_1 + \rho \beta \omega,
\end{split}
\end{equation}
The constants $\alpha$ and $\beta$ are  proportional to the circulation $\kappa$ and depend of the position and orientation of the
body axes. They are explicitly given by:
\begin{equation}
\label{E:circ_constants}
\alpha= \oint_{\partial \B}X\nabla \phi_0 \cdot  \, d\mathbf{l}, \qquad \beta= \oint_{\partial \B}Y\nabla \phi_0 \cdot  \, d\mathbf{l},
\end{equation}
where, as before,  $(X,Y)$ are body coordinates  for material points in $\partial \mathcal{B}$.  The Chaplygin-Lamb equations are discussed in detail in \cite{Borisov_circulation} and a derivation from first principles, using techniques from symplectic geometry and reduction theory, is given in \cite{Joris_circulation}.

  One easily verifies that if the center of the body axes is displaced to
the point with body coordinates $(r,s)$, so that the new body coordinates are
$\tilde X=X-r, \; \tilde Y=Y-s$, then the circulation  constants relative to the new coordinate axes take
the form $\tilde \alpha=\alpha -r\kappa, \; \tilde \beta =\beta -s\kappa$. Thus, there is a unique choice
 of the body axes  that makes these constants vanish. On the other hand,
it is also desired to choose the body axes so that the total inertia tensor $\I$ is diagonal.
For an asymmetric body, these two choices are in general incompatible,  see e.g. \cite{Lamb}.

 For our purposes, the choice of body axes will be made  to simplify the expression of the nonholonomic
 constraint. We therefore consider  equations \eqref{E:Kirchhoff_1_2D_circulation} in their full
generality where $\alpha, \, \beta\neq 0$, and $\I$ is not  diagonal. This contrasts with
the treatment in \cite{Joris_circulation} where it is assumed that $\alpha=\beta=0$ and
with \cite{Borisov_circulation} where the complementary assumption, namely that $\I$ is diagonal,
is made.

It is shown in \cite{GN-V} that equations \eqref{E:Kirchhoff_1_2D_circulation} are of Euler-Poincar\'e type
on a central extension of $SE(2)$ and thus are Hamiltonian.

\section{The hydrodynamic planar Chaplygin sleigh with circulation}
\label{S:Hydro-sleigh-circ}

We are now ready to consider the mechanical system of our interest
which is the generalization of the  hydrodynamic version of the Chaplygin sleigh treated in \cite{hydro-sleigh} to the case when there is circulation around the body. Recall that the classical Chaplygin sleigh problem (going back to 1911, \cite{Chaplygin}) describes the motion of a planar rigid body with a knife edge (a blade) sliding on a horizontal plane. The nonholonomic constraint forbids the motion in the direction perpendicular to the blade. In its hydrodynamic version, the body is surrounded by a potential fluid and the nonholonomic constraint models the effect of a very effective fin or keel, see  \cite{hydro-sleigh}.

With the notation from section~\ref{S:Preliminaries}, we let $\{{\bf E}_1, {\bf E}_2 \}$ be a body frame located at the contact point of the sleigh and the
plane, and so that the ${\bf E}_1$-axis is aligned with the blade (see Figure \ref{F:SleighDiagram}).  The resulting nonholonomic constraint is given by $v_2=0$, and is clearly left invariant under the action of  $SE(2)$, as it is solely written in terms of the velocity of the body as seen in the body frame.

In the absence of constraints, the motion of the body  is described by the Chaplygin--Lamb equations \eqref{E:Kirchhoff_1_2D_circulation}.
In agreement with the Lagrange-D'Alembert principle, which states that the constraint forces perform no work during the
motion, the equations of motion for the constrained system are
%
\begin{equation}
\begin{split}
\label{E:Constrained_2D_circulation}
\dot k &= v_2p_1-v_1p_2 - \rho(\alpha v_1 + \beta v_2), \\
\dot p_1 &= \omega p_2 - \kappa \rho v_2 + \rho \alpha \omega, \\
\dot p_2 &= - \omega p_1 + \kappa \rho v_1 + \rho \beta \omega + \lambda,
\end{split}
\end{equation}
where the multiplier $\lambda$ is determined from the condition $v_2=0$. These equations have been shown to be of
Euler-Poincar\'e-Suslov type on the dual Lie algebra of a central extension of $SE(2)$ in \cite{GN-V}.

\paragraph{The total inertia tensor and the circulation constants $\alpha, \beta$.}
The behavior of the solutions of \eqref{E:Constrained_2D_circulation} depends crucially on the
total inertia tensor $\I=\I_\B+\I_\F$ that relates $(k,p_1,p_2)$ to $(\omega, v_1,v_2)$, and on the
value of the circulation constants $\kappa$ and $\alpha, \beta$.

The expression for  $\I_\B$ with respect to the body frame $\{ {\bf E}_1,  {\bf E}_2 \}$ is
 given by \eqref{E:Inertia_Matrices_2D}
where  $m$ is the mass of the body, $(a,b)$ are body coordinates of the center of mass, and $\In$ is the moment of inertia of the body about the center of mass. While this simple expression is independent
of the body shape, an explicit expression for the tensor of adjoint masses $\I_\F$ can be given explicitly only for rather simple geometries.
A simple yet interesting one is that for an elliptical uniform planar body with the semi-axes of length  $A>B>0$.
Assume  that the origin has coordinates $(r,s)$ with respect to the frame that is aligned
with the principal axes of the ellipse and
that the coordinate axes ${\bf E}_1 \, {\bf E}_2$ are not aligned with the axes of the ellipse,
forming an angle $\upsilon$ (measured counter-clockwise), as illustrated in Figure \ref{F:SleighDiagram}.

\begin{figure}[ht]
\centering
\includegraphics[width=8cm]{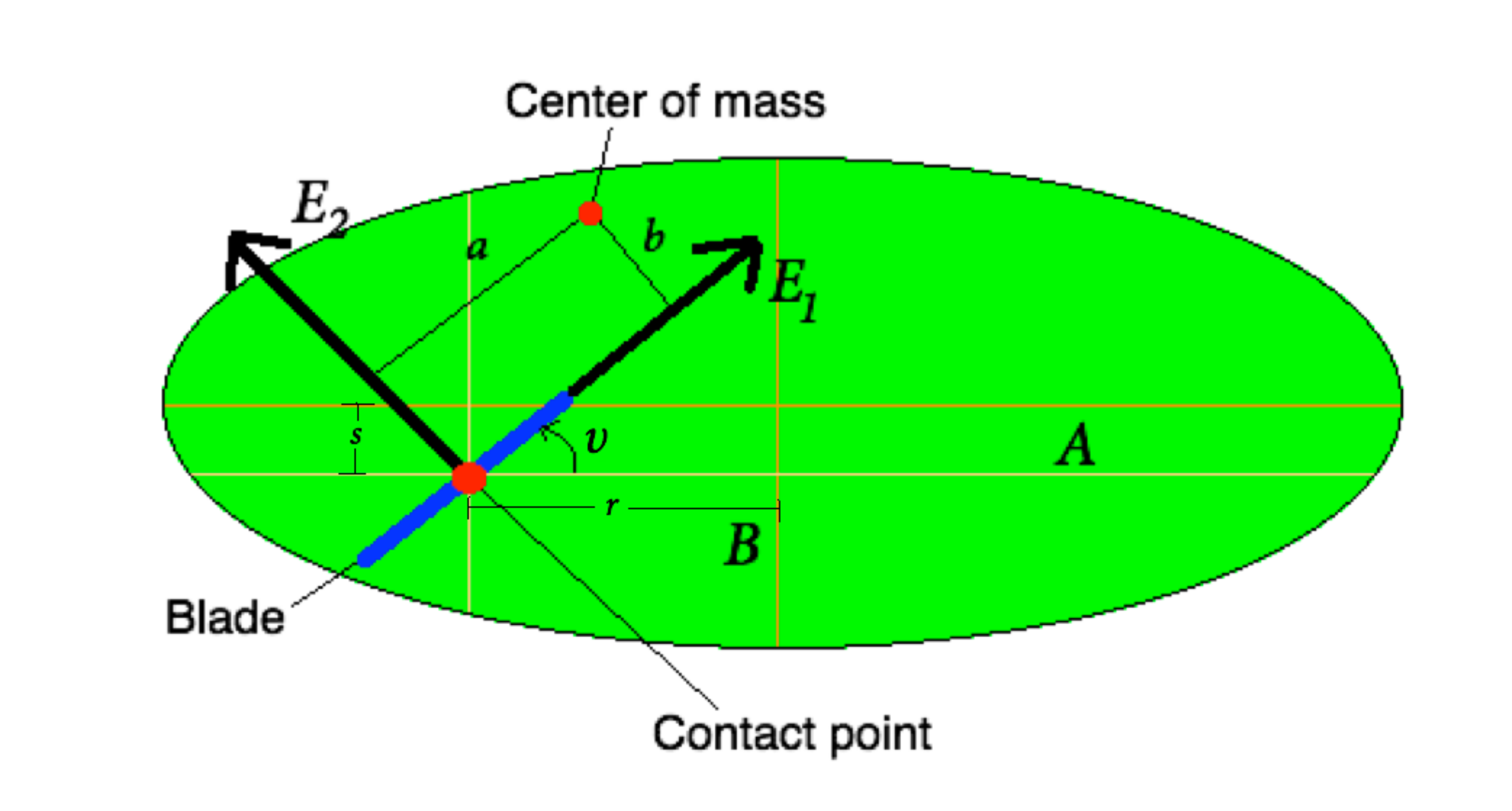}
\caption{\small{The elliptical  sleigh. The blade makes
an angle $\upsilon$ with the major axis of the ellipse and the contact point has coordinates $(r,s)$ with respect to
the frame that is determined by the principal axes of the ellipse (in the diagram both $r$ and $s$ are
negative).\label{F:SleighDiagram}}}
\end{figure}

%

For this geometry, 
starting from the formula \eqref{E:Added_Inertia_Ellipse} for the added inertia tensor given in section \ref{S:Kirchhoff},
one can show that
\begin{equation}
\begin{tiny}
\label{E:Added_Inertia_Ellipse_Rot_Axes}
\I_\mathcal{F}=\rho \pi \left ( \begin{array}{ccc} \begin{array}{c} \frac{ (A^2-B^2)^2}{4}+s^2(B^2\cos^2 \upsilon+ A^2 \sin ^2 \upsilon)  \\ + r^2(A^2\cos^2 \upsilon+ B^2 \sin ^2 \upsilon)  -rs(A^2-B^2) \sin (2\upsilon) \end{array} & \begin{array}{c} s (B^2\cos^2 \upsilon+ A^2 \sin ^2 \upsilon) \\ - \frac 12 r(A^2-B^2) \sin (2\upsilon) \end{array}  & \begin{array}{c} - r (A^2\cos^2 \upsilon+ B^2 \sin ^2 \upsilon) \\ + \frac 12 s(A^2-B^2) \sin (2\upsilon) \end{array} \\ & &  \\  \begin{array}{c} s (B^2\cos^2 \upsilon+ A^2 \sin ^2 \upsilon) \\ - \frac 12 r(A^2-B^2) \sin (2\upsilon) \end{array}  &   B^2\cos^2\upsilon +A^2 \sin^2 \upsilon& \frac{ A^2-B^2}{2} \, \sin (2\upsilon) \\ & & \\  \begin{array}{c} - r (A^2\cos^2 \upsilon+ B^2 \sin ^2 \upsilon) \\ + \frac 12 s(A^2-B^2) \sin (2\upsilon) \end{array} & \frac{ A^2-B^2}{2} \, \sin (2\upsilon) & A^2\cos ^2\upsilon + B^2\sin^2\upsilon \end{array} \right ).
\end{tiny}
\end{equation}

The total inertia tensor, $\I= \I_\mathcal{B} +\I_\F$, of the fluid-body system is then given by
\begin{tiny}
\begin{equation*}
\label{E:Total_Inertia_Ellipse_Rot_Axes}
\I= \left ( \begin{array}{ccc} \begin{array}{c}   \In+m(a^2+b^2)  \\ + \rho \pi  [\frac{ (A^2-B^2)^2}{4}+s^2(B^2\cos^2 \upsilon+ A^2 \sin ^2 \upsilon)  \\ + r^2(A^2\cos^2 \upsilon+ B^2 \sin ^2 \upsilon)  -rs(A^2-B^2) \sin (2\upsilon) ]  \end{array} & \begin{array}{c} -mb + \rho \pi [ s (B^2\cos^2 \upsilon+ A^2 \sin ^2 \upsilon) \\ - \frac 12 r(A^2-B^2) \sin (2\upsilon) ] \end{array}  & \begin{array}{c} ma + \rho \pi [- r (A^2\cos^2 \upsilon+ B^2 \sin ^2 \upsilon) \\ + \frac 12 s(A^2-B^2) \sin (2\upsilon) ] \end{array} \\ & &  \\  \begin{array}{c} -mb + \rho \pi [ s (B^2\cos^2 \upsilon+ A^2 \sin ^2 \upsilon) \\ - \frac 12 r(A^2-B^2) \sin (2\upsilon) ]  \end{array}  &  m+ \rho \pi [ B^2\cos^2\upsilon +A^2 \sin^2 \upsilon ] & \rho \pi \left [ \frac{ A^2-B^2}{2} \, \sin (2\upsilon) \right ] \\ & & \\  \begin{array}{c} ma + \rho \pi [- r (A^2\cos^2 \upsilon+ B^2 \sin ^2 \upsilon) \\ + \frac 12 s(A^2-B^2) \sin (2\upsilon) ]  \end{array} & \rho \pi \left [ \frac{ A^2-B^2}{2} \, \sin (2\upsilon) \right ] & m+ \rho \pi ( A^2\cos ^2\upsilon + B^2\sin^2\upsilon ) \end{array} \right ).
\end{equation*}
\end{tiny}

Notice that in the presence of the fluid, if $\upsilon \neq n\frac{\pi}{2}, n\in \mathbb{Z}$, the coefficient $\I_{23}=\I_{32}$ is non-zero. This can never
be the case if the sleigh is moving in vacuum as one can see from the expression given for $\I_\B$ in  (\ref{E:Inertia_Matrices_2D}).
The appearance of this non-zero term leads to interesting dynamics that, in the absence of circulation were studied in \cite{hydro-sleigh}.
For the above geometry, the circulation constants $\alpha$ and $\beta$ defined by \eqref{E:circ_constants} are computed to be:
\begin{equation*}
\alpha=\kappa ( r\cos \upsilon+ s\sin \upsilon), \qquad \beta = \kappa(  -r\sin \upsilon+ s\cos \upsilon ).
\end{equation*}
Notice that, in the presence of circulation, the two constants, $\alpha$ and $\beta$, can vanish simultaneously only
 if $r=s=0$, that is, only if the contact point is at the center of the ellipse.

In the sequel we assume that the shape of the sleigh is arbitrary convex and that its center of mass  does
not necessarily coincide with the origin, which leads to the general total inertia tensor
\begin{equation*}
\label{E:Total_Inertia_Ellipse_Rot_Axes}
\I= \left ( \begin{array}{ccc} J & -L_2 & L_1 \\ -L_2 & M &Z  \\ L_1& Z & N \end{array} \right ) ,
\end{equation*}
and with arbitrary circulation constants $\alpha, \, \beta$.

\paragraph{Detailed equations of motion.}
The constraint written in terms of momenta is $v_2 = \mathbb{I}^{-1}(k, p_1,p_2)^T=0$. Differentiating it and using \eqref{E:Constrained_2D_circulation}, we find the multiplier
\begin{equation*}
\label{E:Lagrange_Multiplier_part}
\lambda=-\frac{1}{( { \I^{-1}} )_{33}} \left( \I^{-1} \left ( \begin{array}{c} v_2 p_1 - v_1 p_2 -\rho(\alpha v_1 + \beta v_2) \\ \omega p_2 -\kappa \rho v_2 + \rho \alpha \omega \\ -\omega p_1 + \kappa \rho v_1 + \rho \beta \omega) \end{array} \right ) \right )_3,
\end{equation*}
where
\begin{equation*}
\I^{-1}= \frac{1}{\mbox{det}(\I)} \left ( \begin{array}{ccc} MN-Z^2 & ZL_1+NL_2 & -ZL_2-ML_1 \\ ZL_1+NL_2 & JN-L_1^2 &-L_1L_2-JZ  \\ -ZL_2-ML_1 & -L_1L_2-JZ  & JM-L_2^2 \end{array} \right ).
\end{equation*}

A long but straightforward calculation shows that,
by expressing $\omega, v_1$ and $v_2$ in terms of  $k, p_1, p_2$, substituting into \eqref{E:Constrained_2D_circulation},
and  enforcing the constraint $v_2=0$, one obtains:
\begin{equation}
\label{E:Working_Hydro_Sleigh_Equations}
\begin{split}
\dot \omega &=\frac{1}{D}\left (L_1 \omega + Z v_1 + \rho \alpha \right  ) \left ( L_2 \omega - Mv_1\right ),  \\
\dot v_1 &=\frac{1}{D} \left (L_1 \omega + Z v_1  + \rho \alpha \right ) \left (J\omega -L_2 v_1 \right ),
\end{split}
\end{equation}
where we set $D= \operatorname{det} (\I)( {\I^{-1}} )_{33}= MJ-L_2^2$. Note that $D>0$ since $\I$ and $\I^{-1}$ are positive definite.
Note as well that if $\alpha=0$ we recover
 the system with zero circulation treated in \cite{hydro-sleigh} so from now on we assume $\alpha \neq 0$.

The full motion of the sleigh on the plane is determined by the reconstruction equations
which, in our case with $v_2=0$, reduce to
\begin{equation*} \label{recon}
\dot \theta = \omega, \qquad \dot x=v_1\cos \theta, \qquad \dot y = v_1\sin \theta.
\end{equation*}

The reduced energy integral has
\begin{equation*}
H=\frac{1}{2}\left ( J\omega^2 +Mv_1^2 -2L_2\omega v_1 \right ),
\end{equation*}
and its level sets are ellipses on the $(\omega \, v_1)$-plane.

As seen from the equations,
the straight line $\ell=\{L_1 \omega + Z v_1 + \rho  \alpha=0\}$ consists of equilibrium points for the system.
Each of these equilibria corresponds to a uniform circular motion on the plane along a
circumference of radius $\left | \frac{v_1}{\omega} \right |$.

 Notice that if $Z=L_1=0$  the line $\ell$ of equilibria disappears.
In fact, it is shown in \cite{GN-V} that equations \eqref{E:Working_Hydro_Sleigh_Equations}
possess an invariant measure only for this choice of the parameters.
 In this particular case we obtain simple harmonic motion on the reduced plane $\omega, v_1$.

For the sequel we will assume that $Z$ and $L_1$ are not both zero. We shall see that
initial conditions with high energy yield asymptotic dynamics which is in agreement with
our statement  that there is no invariant measure in this case.

A level set of the energy will intersect once, twice or never the line of equilibra $\ell$ depending
on the value of $H$. One can show that there are two intersections if $H>h_0$, only one
intersection if $H=h_0$, and zero intersections if $H<h_0$, where
\begin{equation*}
h_0=\frac{1}{2}(\rho \alpha)^2 \frac{D}{E}
\end{equation*}
with $E=JZ^2+2ZL_1L_2+ML_1^2>0$.

Hence, the trajectories of (\ref{E:Working_Hydro_Sleigh_Equations}) are contained
in ellipses and they are of three types:
\begin{enumerate}
\item For small values of the energy, $0\leq H <h_0$, we have periodic motion on the ellipses.
\item For $H=h_0$ we have a homoclinic connection that separates periodic from asymptotic trajectories.
\item For $H>h_0$ we have heteroclinic connections
between the asymptotically unstable and stable equilibria on $\ell$.
\end{enumerate}

We introduce two other special energy values
$h_1, h_2\geq h_0$, for which the corresponding energy contour line intersects the equilibra
line $\ell$ at the axes $\omega=0$ and $v_1=0$ respectively. Namely,
\begin{equation}
\label{E:h_1h_2}
h_1=\frac{1}{2} \, \frac{(\rho \alpha)^2 M}{Z^2}, \qquad h_2=\frac{1}{2} \, \frac{(\rho \alpha)^2 J}{L_1^2} .
\end{equation}

The phase portrait is illustrated in Figure \ref{F:phase-portrait} (a).

\begin{figure}[ht]
\centering
\subfigure[Reduced phase portrait. The stable and unstable equilibra are represented by filled and empty dots, respectively.]{\includegraphics[width=5.1cm]{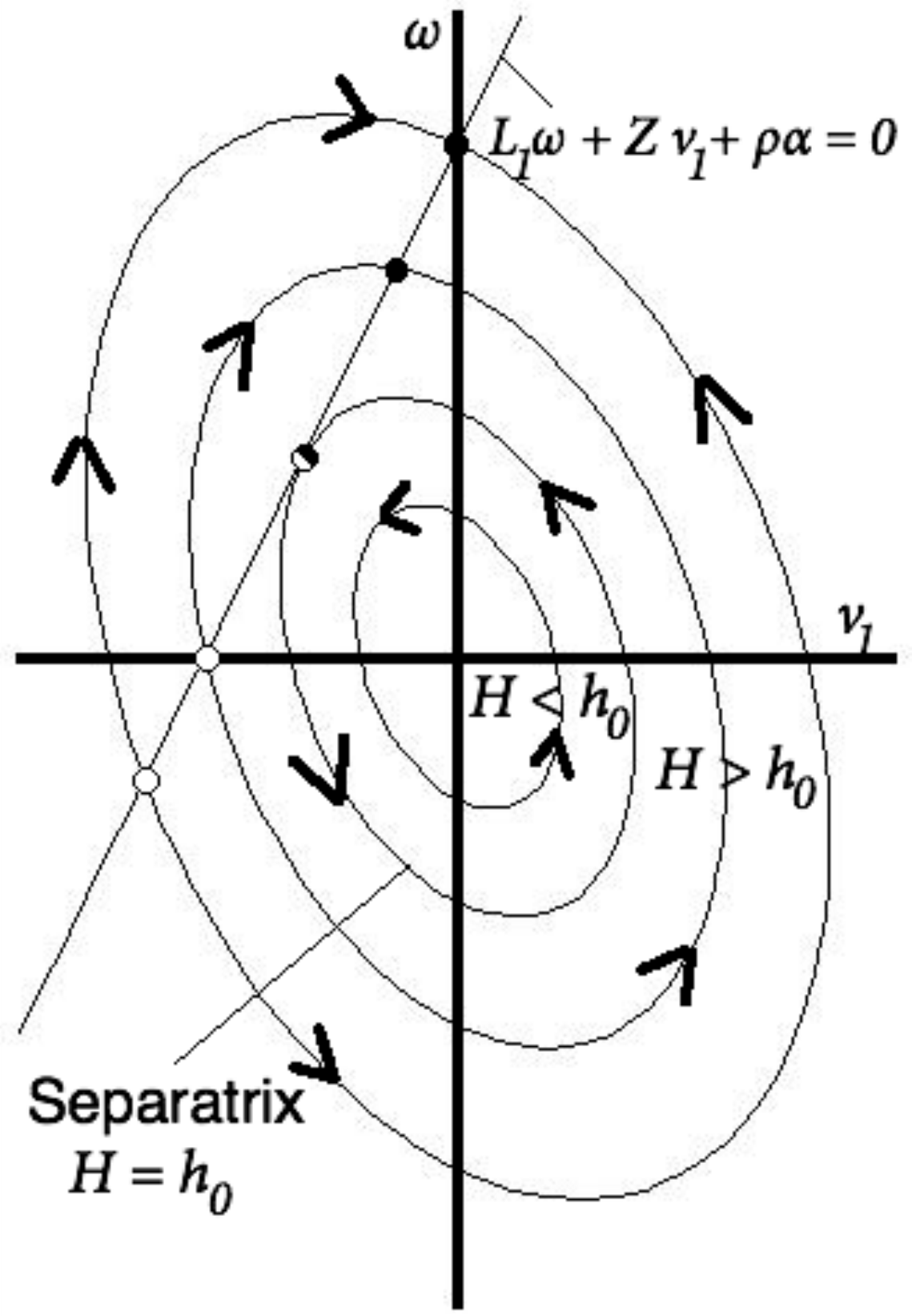}} \qquad \qquad
\subfigure[The 7 regions in the reduced phase space. ]{\includegraphics[width=5.7cm]{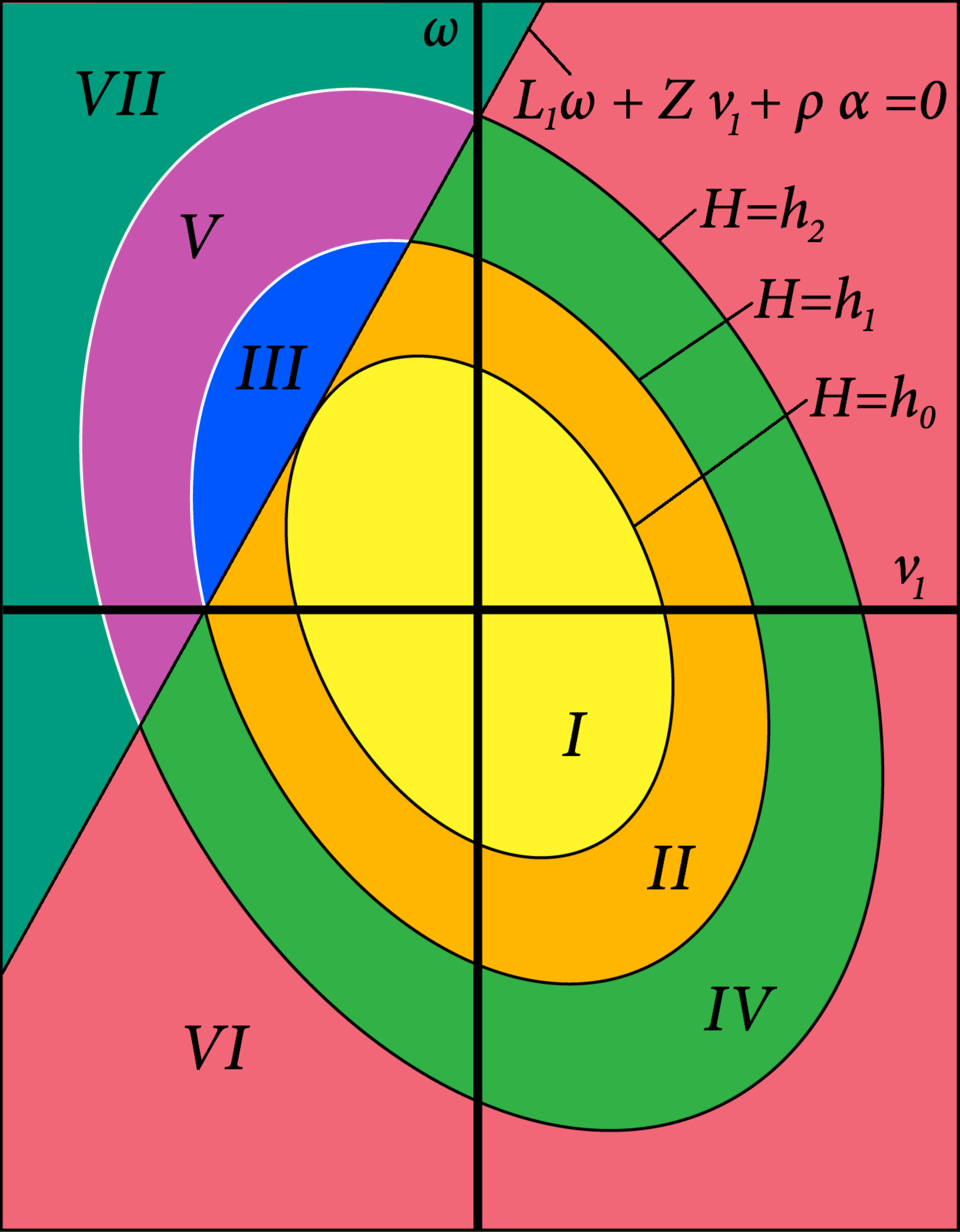}}
\caption{Reduced phase portrait  and the 7 regions in phase space (assuming that  $L_1>0$,  $\alpha, Z <0$, and
$h_1< h_2$).  \label{F:phase-portrait}}
\end{figure}


\paragraph{Remark 1.} In fact, the reduced system (\ref{E:Working_Hydro_Sleigh_Equations}) can be checked to be Hamiltonian
with respect to the following Poisson bracket of functions of $\omega, v_1$
\begin{equation*}
\{F_1,F_2 \}:= -\frac{1}{D}\left (L_1 \omega + Z v_1 + \rho \alpha \right ) \left ( \frac{\partial F_1}{\partial \omega} \frac{\partial F_2}{\partial v_1} -
\frac{\partial F_1}{\partial v_1} \frac{\partial F_2}{\partial \omega} \right ).
\end{equation*}
The invariant symplectic leaves  consist of the semi-planes separated by the equilibria line $\ell$
and the zero-dimensional leaves formed by the points on $\ell$.
\medskip

\subsection{The motion of the sleigh on the plane: qualitative description}

As for the reduced dynamics, the qualitative motion of the sleigh  on the plane depends crucially on the energy value.
We distinguish 7 regions on the reduced phase space as illustrated in Figure \ref{F:phase-portrait} (b). The regions
 are separated by the homoclinic orbit corresponding to the energy contour $H=h_0$, the two pairs of heteroclinic orbits
 corresponding to the energy contour lines $H=h_1$ and $H=h_2$, and the line of equilibra $\ell$. It is thus
 clear that the regions are invariant by the flow.
 We shall see that
the qualitative motion of the sleigh is different in each region.
 The number of regions can change
for special parameter values that make any of $h_0, h_1, h_2$ coincide.
We shall discuss the behavior
in the interior of the regions in the generic case where they are all different and under the assumption that $h_1<h_2$.

\paragraph{Motion in Region I.} For subcritical values of the energy, $0\leq H <h_0$,
the reduced dynamics on $(v_1,\omega)$ is periodic, but
the motion on the sleigh is generally not periodic. Let $T$ be the (minimal) period of the reduced dynamics
that depends on the energy value $H$ in a way that will be made precise in section \ref{S:solution}.
Under the shift $t \to t+T$ 
the coordinates of the sleigh undergo the increments
\begin{equation}
\label{E:Def_Phase_Shift}
\Delta x=\int_0^Tv_1(t)\cos\theta(t)\, dt, \qquad \Delta y=\int_0^Tv_1(t)\sin \theta(t)\, dt, \qquad
\Delta \theta =\int_0^T\omega(t)\, dt.
\end{equation}
These quantities completely determine the type of motion of the sleigh.

\begin{theorem}
\label{T:Reconstruction_periodic}
 Let $T$ be the period of a periodic solution to the reduced system \eqref{E:Working_Hydro_Sleigh_Equations} with $H<h_0$,
 and let ${\bf \Delta x}=(\Delta x, \Delta y)$ and $\Delta \theta$ be defined
by \eqref{E:Def_Phase_Shift}. Then, the motion of the sleigh on the plane is
\begin{enumerate}
\item $T$-periodic if $\frac{\Delta \theta}{2\pi}\in \mathbb{Z}$ and ${\bf \Delta x}=0$.
\item Unbounded if $\frac{\Delta \theta}{2\pi}\in \mathbb{Z}$ and ${\bf \Delta x}\neq 0$.
\item Contained in a circle or an  annulus if  $\frac{\Delta \theta}{2\pi}\notin \mathbb{Z}$, where the motion
\begin{enumerate}
\item is periodic of period $qT$ if $\frac{\Delta \theta}{2\pi}$ is a rational number written in
terms of integers $p,q$ as $\frac{\Delta \theta}{2\pi}=\frac{p}{q}$ in irreducible form,
\item fills up the annulus densely if $\frac{\Delta \theta}{2\pi}$ is irrational.
\end{enumerate}
\end{enumerate}
Moreover, the above types of the motion depend only on the energy value $H$. In particular, they
are independent of the initial configuration of the sleigh.
\end{theorem}

\begin{proof} Assume without loss of generality that at $t=0$ the body and the space reference frames coincide.
Then $\theta(0)=0$, and the element
$$
S =\left (
\begin{array}{ccc}
\cos (\Delta\theta) & -\sin (\Delta\theta) & \Delta x \\
\sin (\Delta\theta) &\cos (\Delta\theta) & \Delta y \\
0 & 0 & 1
\end{array}\right ) \in SE(2)
$$
describes the shift in both reference frames after one period. Then the elements
\begin{equation*}
g_n =\left (
\begin{array}{ccc}
\cos \theta (n T) & -\sin \theta (nT) & x(nT) \\
\sin \theta (nT) &\cos \theta (nT) & y(nT) \\
0 & 0 & 1
\end{array}\right )\in SE(2), \qquad n\in {\mathbb Z}
\end{equation*}
describing the subsequent positions of the sleigh in the space frame, are the right translations on $SE(2)$:
$$
g_{n+1} = g_n S.
$$
Following theorems of kinematics,
each generic transformation of that kind is a rigid rotation about
a {\it fixed} point $Q$ on the plane $(x,y)$ by the angle $\Delta\theta$. This happens when
$\frac{\Delta \theta}{2\pi} \notin \mathbb{Z}$.
Then the trajectory of the contact point lies inside of an annulus or a circle
with the center $Q$.
The trajectory is periodic iff $\frac{\Delta \theta}{2\pi}$ is a rational number, otherwise it fills the domain
densely.

In the special case $\frac{\Delta \theta}{2\pi} \in \mathbb{Z}$ the transformation $g_n\to g_{n+1}$
is a parallel translation
(if ${\bf \Delta x}\ne 0$), which is obviously unbounded, or is the identity (when ${\bf \Delta x}= 0$), then the
trajectory is periodic. This leads to the items of the theorem.

Finally note that the components of $S$ depend only on the energy value and not on the particular initial conditions.
\end{proof}

We mention that the matrix $S$ in the proof is the (right) monodromy matrix of the periodic linear system defined
by the reconstruction equations. The global description of the invariant manifolds in the unreduced space, that generically are two-tori, can be studied with techniques
similar to those developed in
\cite{fasso}.

The types of motion indicated in items 2 and 3 of the above theorem are illustrated in Figure \ref{F:Trajectory}.
The generic behavior of the sleigh  corresponds to  item 3 (b) (Figure \ref{F:Trajectory} (c)).

Equation \eqref{E:Deltaphi} in section \ref{S:solution} below gives the explicit dependence of
$\Delta \theta$ on the energy.

 \begin{figure}[ht]
\centering
\subfigure[Trajectory of the contact point with $\Delta \theta=-4\pi$, and $\Delta x, \Delta y \neq 0$ (unbounded
motion).]{\includegraphics[width=3.5cm]{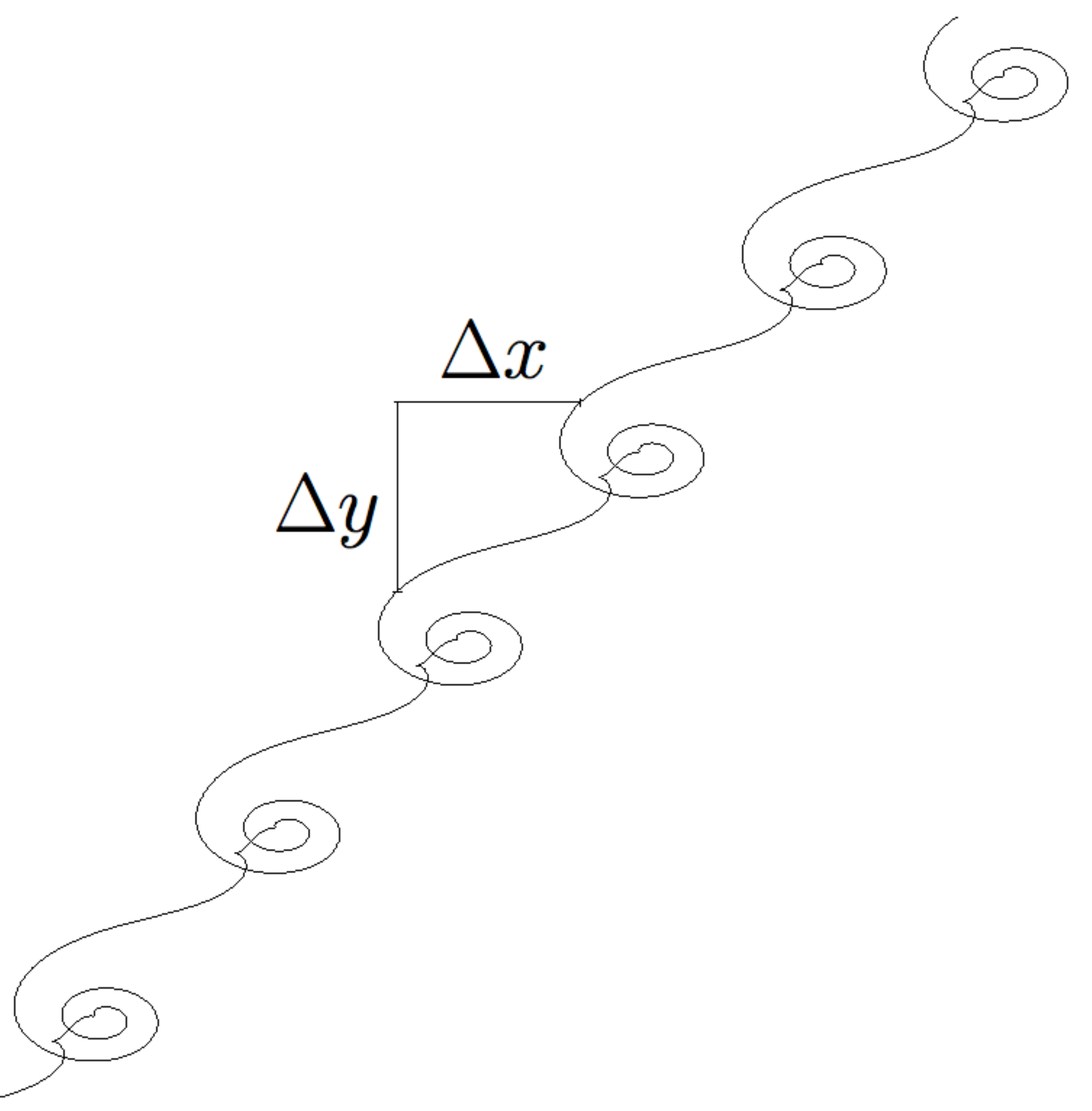}} \qquad
\subfigure[Trajectory of the contact point with $\Delta \theta = -\frac{11\pi}{7}$ (periodic motion in an annulus).]{\includegraphics[width=4cm]{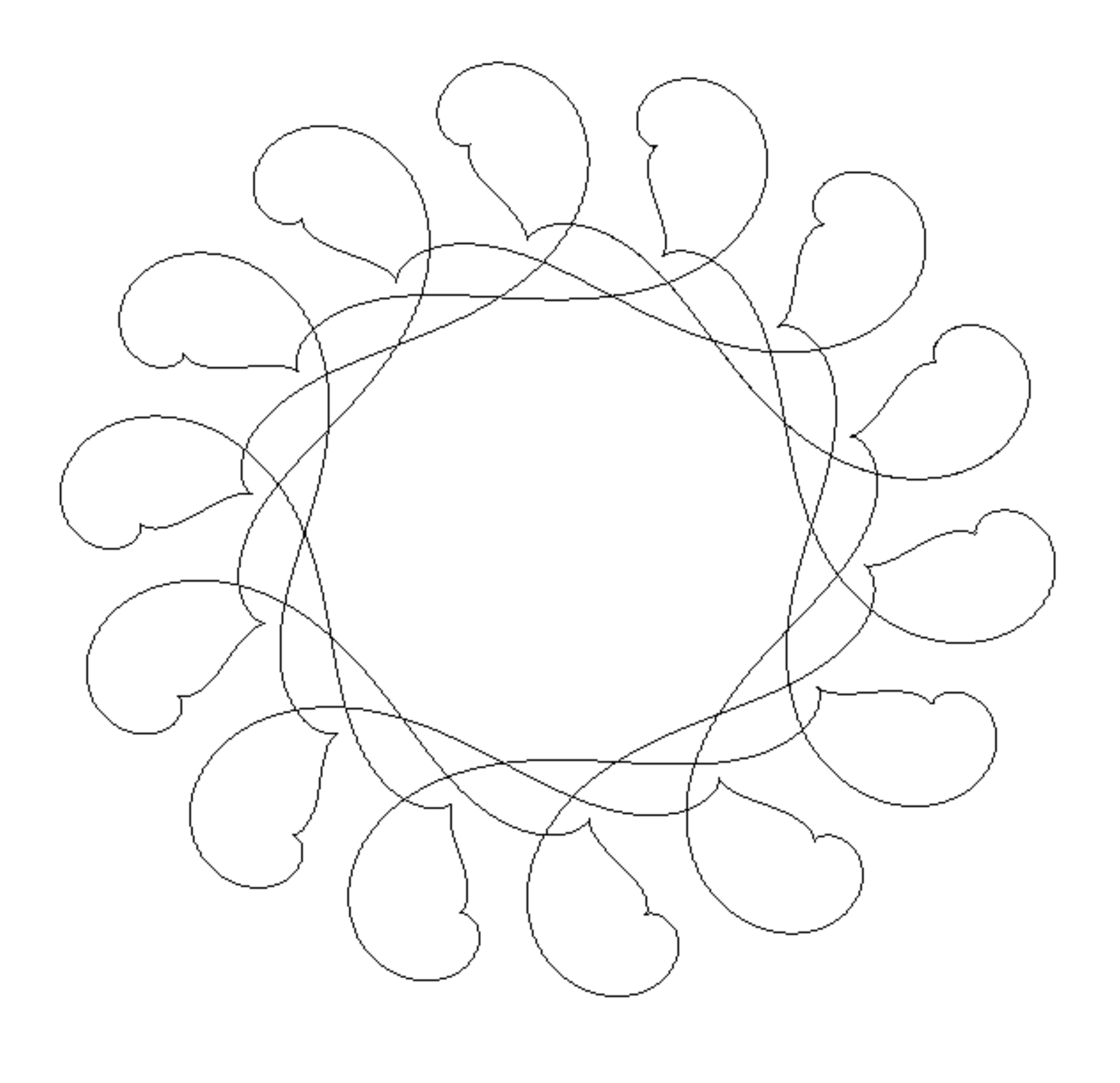}} \qquad
\subfigure[Trajectory of the contact point with $\Delta \theta = -\frac{e\pi}{2}$ for a finite time (the trajectory fills the annulus densely).]{\includegraphics[width=4cm]{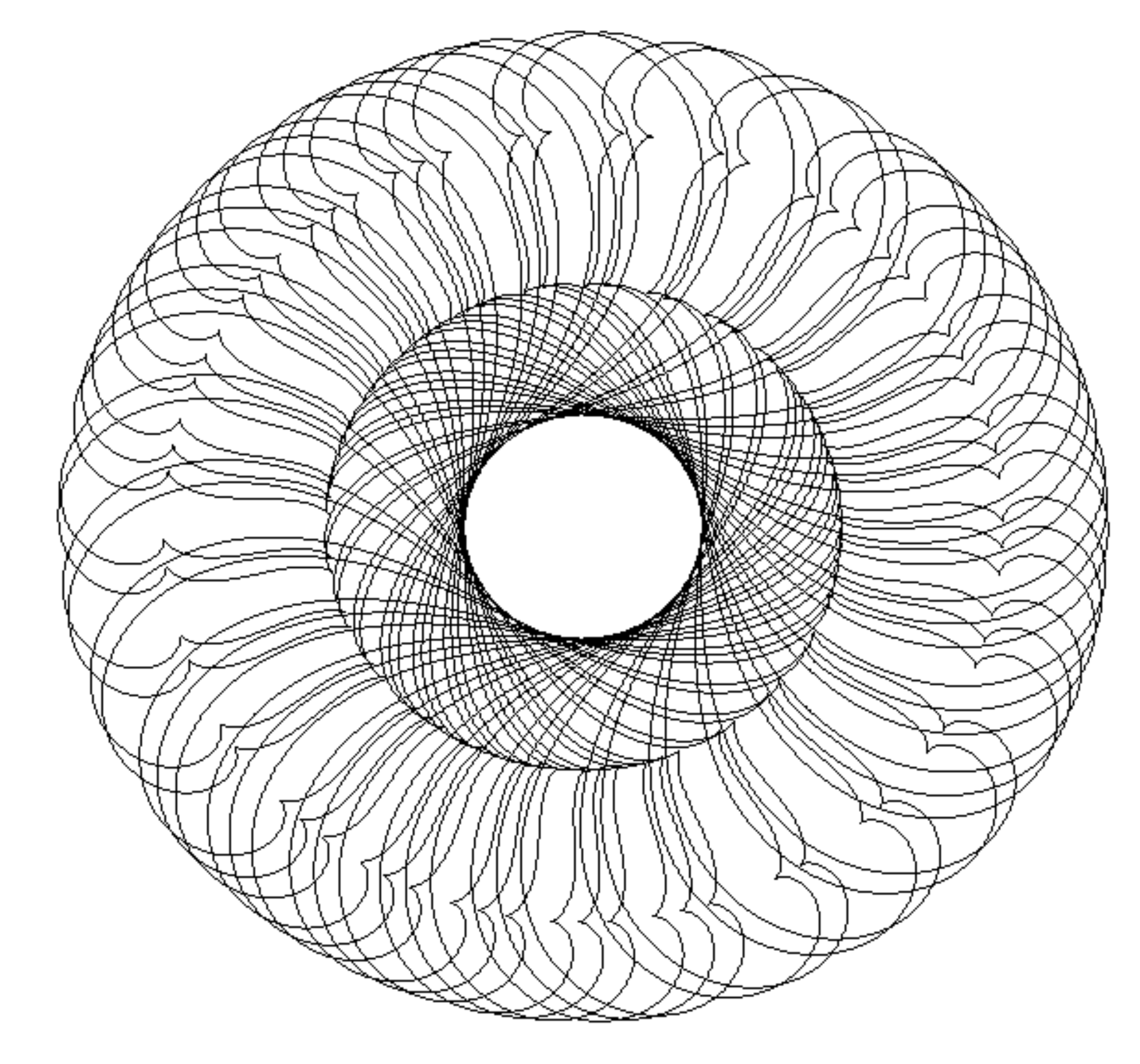}}
\caption{Trajectory of the sleigh on the plane in region I ($H<h_0$) for different values of $\Delta \theta$.  \label{F:Trajectory}}
\end{figure}

Now notice that in the case $h\ge h_0$, the equilibria points $(\omega, v_1)$
on the line $\ell=\{L_1\omega +Zv_1+\rho \alpha =0\}$ correspond to circular motion of the sleigh
in the direction determined by the sign of $\omega$ (clockwise if $\omega >0$ and anti-clockwise if  $\omega <0$).
The circles describing by the contact point have radius $|v_1/\omega|$.

The exceptional cases when $\omega =0$ or $v_1=0$ respectively correspond to motion along a straight line
or to a spinning motion of the sleigh about the contact point that remains fixed.
These cases correspond respectively to the  energy values $h_1$ and $h_2$. The motion of the
sleigh on the plane when the energy attains the critical value $h_0$ will be discussed in section \ref{S:solution}.

\paragraph{Motion in Region II.} The sleigh evolves asymptotically from one circular motion to another. The
limit circles have different radii. Both of them are traversed counterclockwise (at the limit points in the region one has $\omega>0$). Moreover, the trajectory of the point of contact has two cusps corresponding
to the two intersections with the axes $v_1=0$ on the Figure \ref{F:phase-portrait}.

\paragraph{Motion in Region III.} The same behavior as in Region II, however, this time the
trajectory has no cusps (since $v_1$ never vanishes).

\paragraph{Motion in Region IV.} The sleigh again evolves asymptotically from one circular motion to another. The
limit circles have different radius and are are traversed in {\it opposite} directions.
The trajectory of the contact point has two cusps corresponding to the two intersections with the axes $v_1=0$ on the
plane $(v_1,\omega)$.

\paragraph{Motion in Region V.} The same behavior as in Region IV, but the
trajectory has no cusps since $v_1$ never vanishes.

\paragraph{Motion in Regions VI and VII.} The same behavior as in Region IV, however
the trajectories have exactly one cusp corresponding to the intersection of the axis $v_1=0$.

\begin{figure}[ht]
\label{F:Trajectory-Regions234}
\centering
\subfigure[Region II.]{\includegraphics[width=4.2cm]{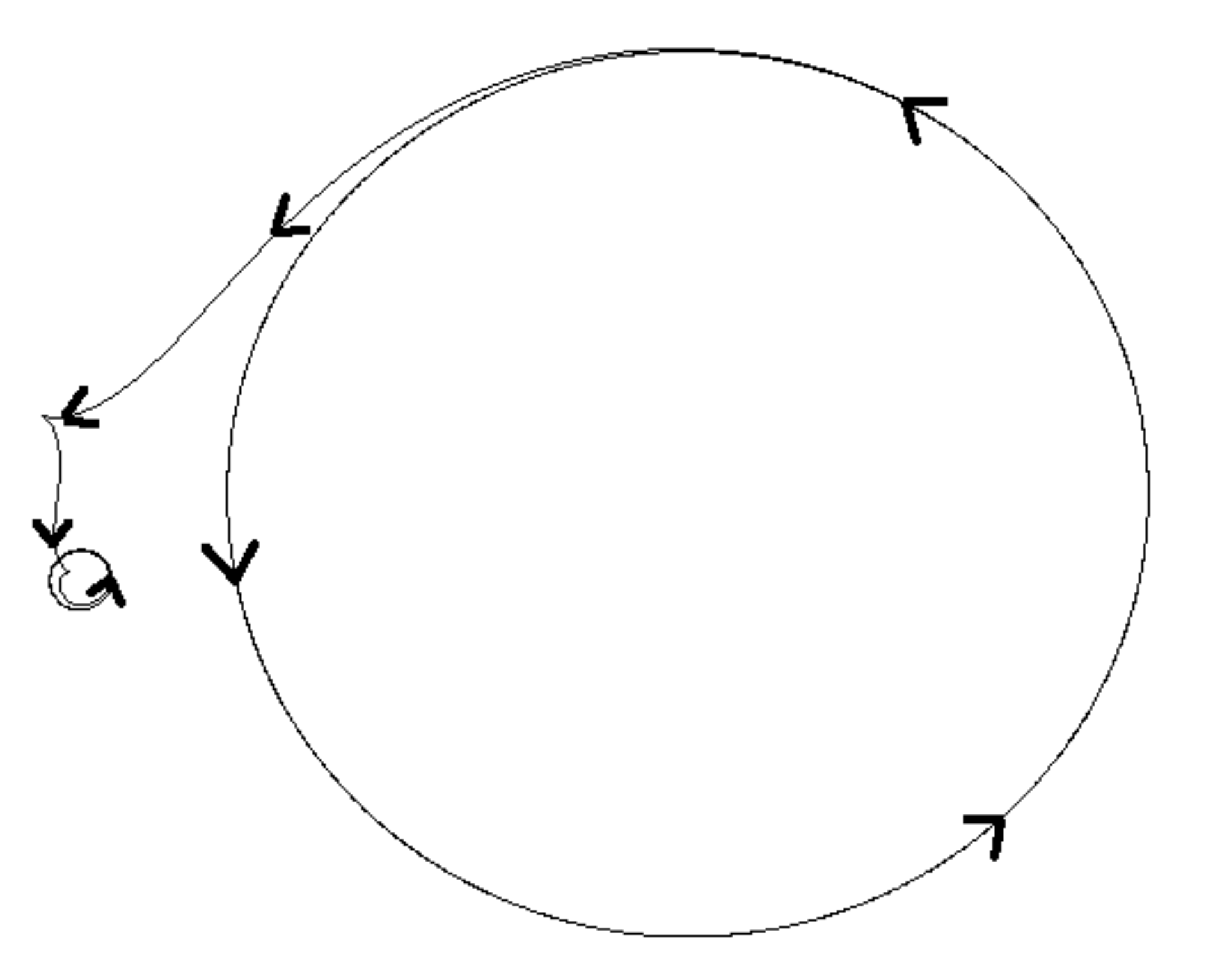}} \qquad
\subfigure[Region III.]{\includegraphics[width=3.5cm]{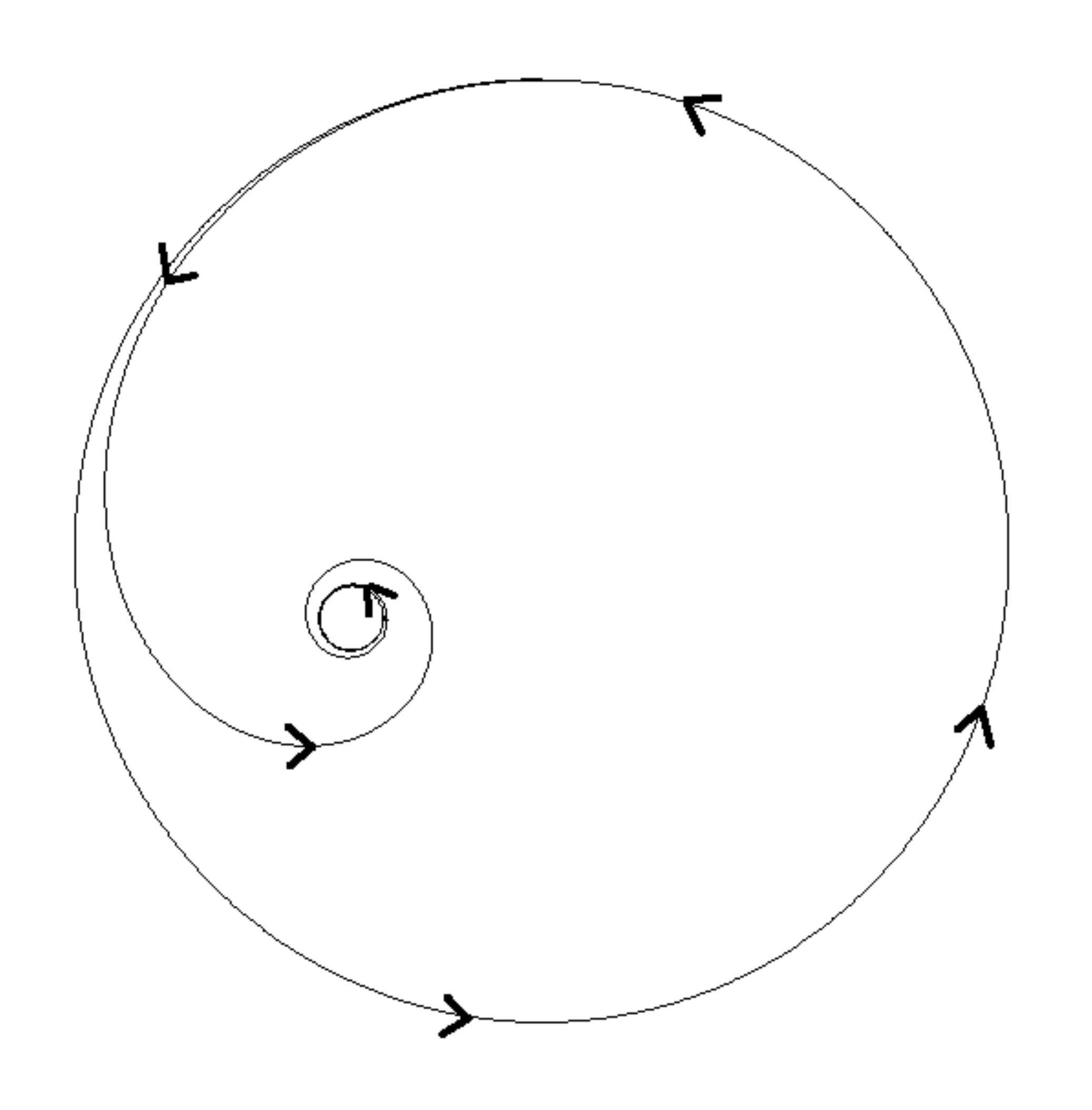}} \qquad
\subfigure[Region IV.]{\includegraphics[width=5.8cm]{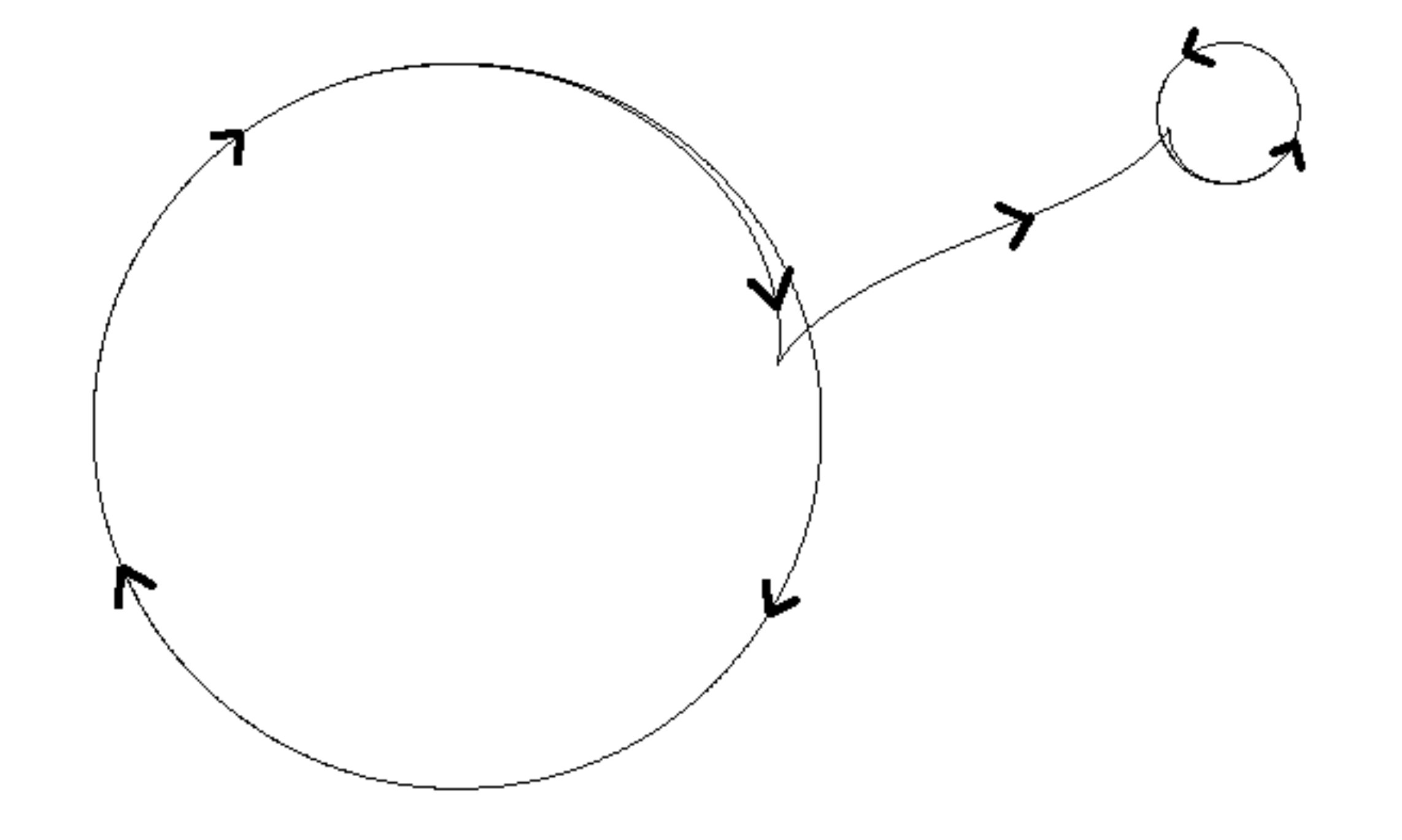}}
\caption{Trajectories of the sleigh for initial conditions in regions II, III, and IV.}
\end{figure}

 \begin{figure}[ht]
 \label{F:Trajectory-Regions234}
\centering
\subfigure[Region V.]{\includegraphics[width=5cm]{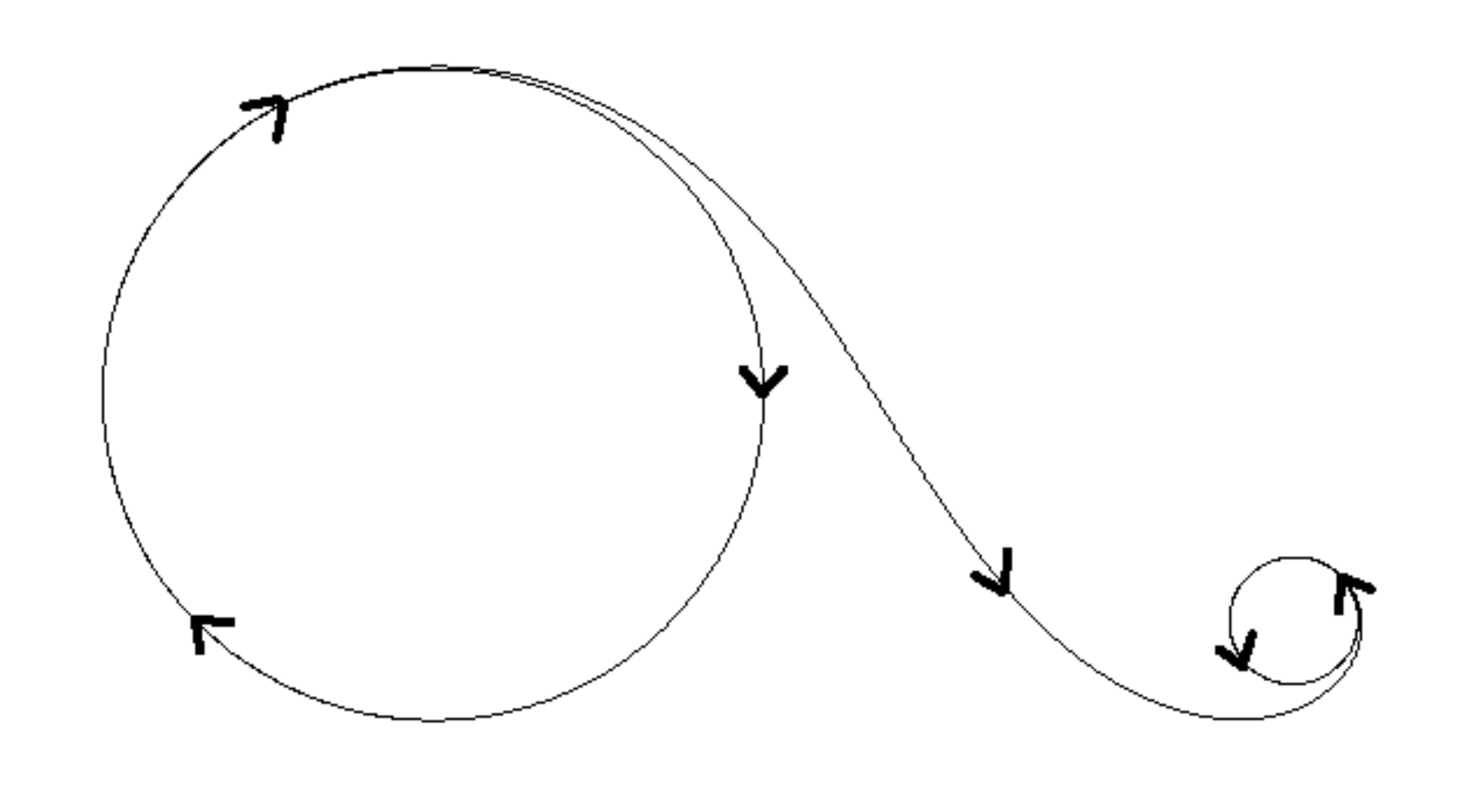}} \qquad
\subfigure[Region VI.]{\includegraphics[width=5cm]{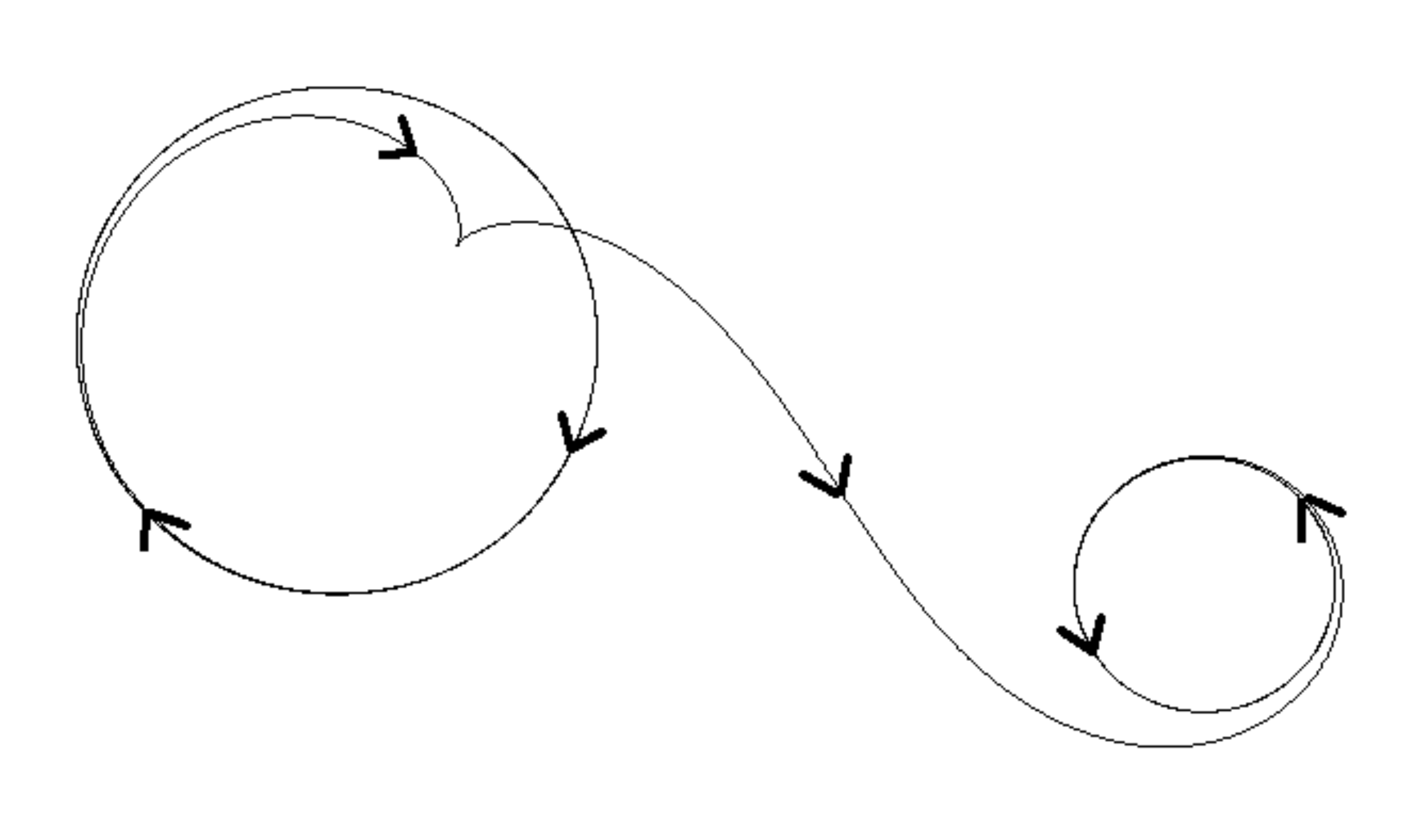}} \qquad
\subfigure[Region VII.]{\includegraphics[width=4.6cm]{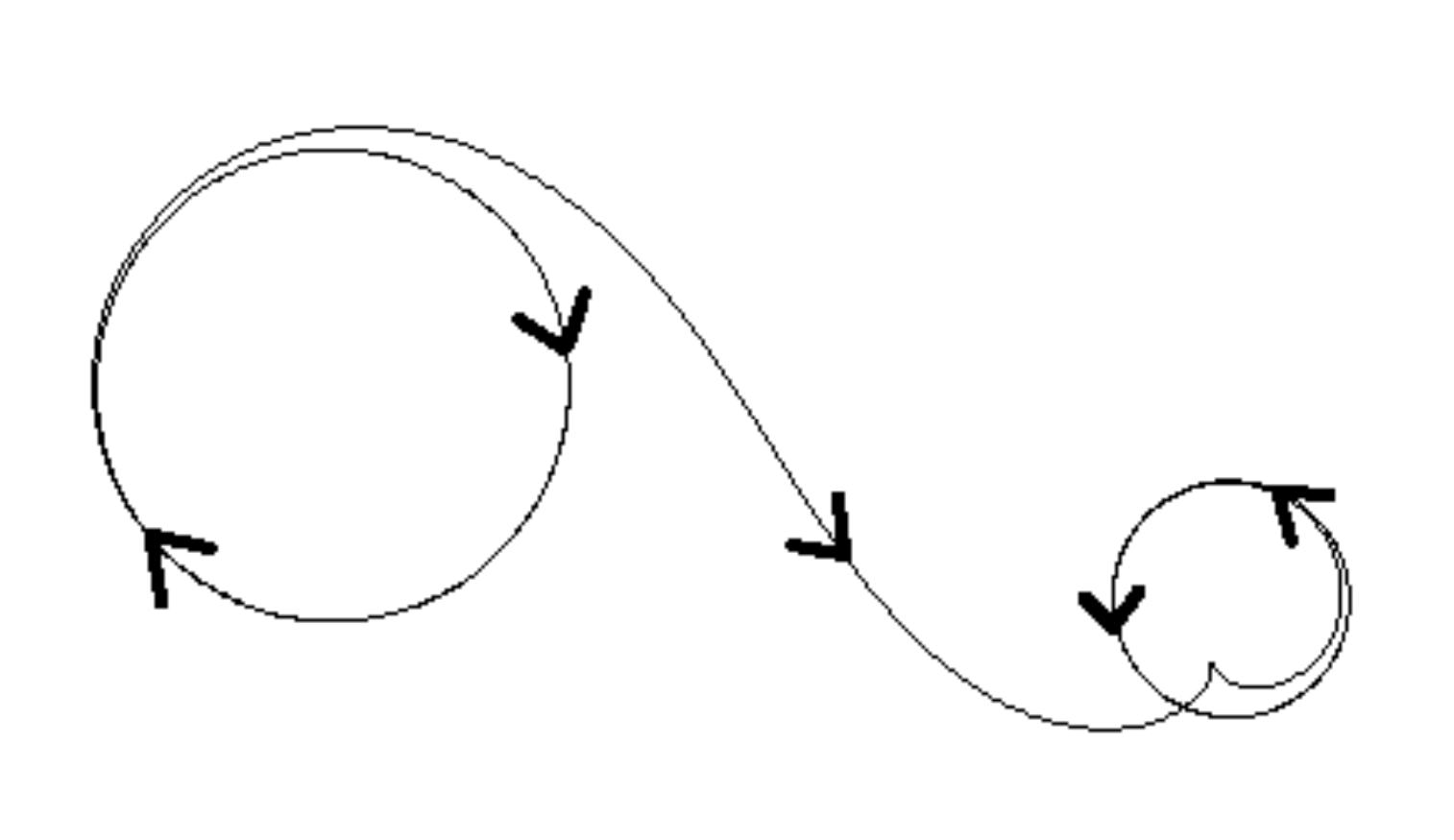}}
\caption{The trajectories for initial conditions in regions V, VI, and VII.}
\end{figure}

Thus we see that the qualitative behavior of the motion of the sleigh on the plane can be sensitive to initial
conditions. This contrasts with the properties of other completely solvable classical nonholonomic systems
like the Suslov problem or the hydrodynamic Chaplygin sleigh without circulation. For these systems
one can obtain  formulas for the parameters that determine the long term behavior of the dynamics, and such
formulas are independent of the energy (see \cite{hydro-sleigh, Suslov-Fedorov}).

At the physical level, the motion of the sleigh can be understood as the interplay of two effects, the circulation of
the fluid around the body, and the inertia of the body. For low energies ($H<h_0$) the dynamics are driven by the circulation
and produce periodic motion in the reduced space. For high energies ($H>h_0$) the inertia of the body takes over and the asymptotic dynamics
of the body on the plane resemble the motion of the hydrodynamic
Chaplygin sleigh in the absence of circulation considered in \cite{hydro-sleigh}.

We finally stress that all the above types of motion also appear in
the classical Appel--Korteweg problem on the rolling of a non-vertical disc. Namely the
trajectories of the contact point of the disc admit the same classification (see \cite{App, LCA1952, Fed87}).

\section{Explicit solution of the reduced system in $\se (2)^*$}
\label{S:solution}

We will only treat the case $L_2=0$ when the system \eqref{E:Working_Hydro_Sleigh_Equations}
takes the form
\begin{equation}
\label{E:Simp_Working_Hydro_Sleigh_Equations}
\begin{split}
\dot \omega &=\frac{1}{D}\left (L_1 \omega + Z v_1 + \rho \alpha \right  ) \left ( - Mv_1\right ),  \\
\dot v_1 &=\frac{1}{D} \left (L_1 \omega + Z v_1  + \rho \alpha \right ) \left (J\omega  \right ),
\end{split}
\end{equation}
with $D=MJ$. The energy takes the diagonal form $H=\frac{1}{2}(J\omega^2 +Mv_1^2)$. The general
case can be reduced to this one via an invertible linear change of variables as the following proposition shows.

\begin{proposition}
There exists a $2\times 2$ rotation matrix $P\in SO(2)$ such that if $\omega$ and $v_1$ satisfy
\eqref{E:Working_Hydro_Sleigh_Equations} then $(\tilde \omega, \tilde v_1)^T:=P( \omega,  v_1)^T$ satisfy
\eqref{E:Simp_Working_Hydro_Sleigh_Equations} with
coefficients $\tilde L_1, \tilde Z, \tilde M, \tilde J$, where
$\tilde J, \tilde M>0$ satisfy  $\tilde M \tilde J =D$, and  $(\tilde L_1, \tilde Z)^T:=P(L_1,Z)^T$.
\end{proposition}
\begin{proof}
Define the restricted inertia tensor $\I_c$ by
\begin{equation*}
\I_c:=\left ( \begin{array}{cc} J & - L_2 \\ -L_2 & M \end{array} \right ).
\end{equation*}
Since $\I_c$ is symmetric and positive definite, there exists $P\in {SO}(2)$ such that
$P\, \I_c \,P^T$ is a diagonal matrix whose entries $\tilde J, \tilde M >0$ are the eigenvalues of $\I_c$ and
satisfy $\tilde M \tilde J =\mbox{det} \, \I_c =D$.

Equations  \eqref{E:Working_Hydro_Sleigh_Equations} can be written in vector form as
\begin{equation*}
\frac{d}{dt} \, \left ( \begin{array}{c} \omega \\ v_1 \end{array} \right ) = \frac{1}{D} \left (
 \left \langle \left ( \begin{array}{c} \omega \\ v_1 \end{array} \right ) \, , \, \left ( \begin{array}{c} L_1 \\ Z \end{array} \right )  \right \rangle + \rho \alpha \right ) \left ( \begin{array}{cc}0 & -1 \\ 1 & 0 \end{array} \right )   \, \I_c \,  \left ( \begin{array}{c} \omega \\ v_1 \end{array} \right ),
\end{equation*}
where $\langle \cdot , \cdot \rangle$ denotes the usual inner product in $\R^2$.
Defining  $(\tilde \omega, \tilde v_1)^T:=P( \omega,  v_1)^T$, the above equations are equivalent to
 \begin{equation*}
\frac{d}{dt} \, \left ( \begin{array}{c} \tilde \omega \\  \tilde v_1 \end{array} \right ) = \frac{1}{D} \left (
 \left \langle \left ( \begin{array}{c}\tilde \omega \\  \tilde v_1 \end{array} \right ) \, , \, P \, \left ( \begin{array}{c} L_1 \\ Z \end{array} \right )  \right \rangle+ \rho \alpha \right ) \, P \, \left ( \begin{array}{cc}0 & -1 \\ 1 & 0 \end{array} \right )  \, \I_c \, P^T \,  \left ( \begin{array}{c} \tilde{\omega} \\ \tilde{v}_1 \end{array} \right ).
\end{equation*}
The result now follows by noticing that
\begin{equation*}
P \, \left ( \begin{array}{cc}0 & -1 \\ 1 & 0 \end{array} \right )  \, \I_c \, P^T =  \left ( \begin{array}{cc}0 & -1 \\ 1 & 0 \end{array} \right ) \, P \, \I_c \, P^T = \left ( \begin{array}{cc}0 & -1 \\ 1 & 0 \end{array} \right ) \left ( \begin{array}{cc}\tilde J & 0 \\ 0 & \tilde M \end{array} \right )
\end{equation*}
where we have used that  $P\in SO(2)$ in the first equality.
\end{proof}

The  solution of the system \eqref{E:Simp_Working_Hydro_Sleigh_Equations} can be obtained
by parameterizing level set $H=h$ of the energy function $H=\frac{1}{2}(J\omega^2 +Mv_1^2)$ as
\begin{equation*}
\omega(t)=\sqrt{ \frac{2h}{J}} \left ( \frac{2f(t)}{f(t)^2+1} \right ), \qquad v_1(t)= \pm \sqrt{\frac{2h}{M}} \left (
\frac{f(t)^2 -1}{f(t)^2+1} \right ).
\end{equation*}
Substitution into \eqref{E:Simp_Working_Hydro_Sleigh_Equations} gives a separable equation for $f(t)$ of
the form
\begin{equation}
\label{E:eqforf}
\dot f(t) = A f(t)^2 +2Bf(t) +C,
\end{equation}
where the coefficients $A, B, C,$ depend on the entries of the inertia tensor $\I$, on $\rho \alpha$, and also  on the energy value
$h$. The form of the solutions for the above equation depends on the sign of the discriminant $\Delta:=B^2-AC$.
A long but straightforward calculation shows that
\begin{equation*}
\Delta = \frac{1}{2}\left ( \frac{E}{D^2} \right ) (h-h_0),
\end{equation*}
and hence $\mbox{sign}(\Delta)=\mbox{sign}(h-h_0)$. Therefore the form of the solution varies depending
on whether the energy $h$ is higher than, smaller than or equal to the critical energy level $h_0$.
We will now give the explicit formulae for $\omega$ and $v_1$ in these three cases. To write our formulas
in a compact and unified manner we introduce  the coefficients:
\begin{equation}
\label{E:def-coeff}
\begin{split}
\gamma & := \sqrt{|\Delta|}=\frac{\sqrt{E}}{\sqrt{2}D} \sqrt{|h-h_0|}, \qquad K_1:=\frac{\sqrt{E}}{Z\sqrt{J}}\left (\frac{\sqrt{|h-h_0|}}{\sqrt{h}+\sqrt{h_1}} \right ), \\ K_2 &:=-\mbox{sign}(\alpha Z) \frac{L_1 \sqrt{M}}{Z\sqrt{J}} \left (
\frac{\sqrt{h}}{\sqrt{h_1}+\sqrt{h}} \right ), \qquad
K_3:=-\frac{\sqrt{2}M\sqrt{J}}{Z}\left ( \frac{1}{\sqrt{h_0} +\sqrt{h_1}} \right ),
\end{split}
\end{equation}
where $h_1$ is given by
\eqref{E:h_1h_2},
and we assume that the constants $L_1, Z$, are non-zero. Notice that  $\gamma, K_1,$ and $K_2$ depend on the energy value $h$.

\paragraph{The solution for $0<h<h_0$.}

In this case the discriminant $\Delta <0$ and  $f(t)=\frac{1}{A}(\gamma \tan (\gamma t)-B)$ is a
 solution of \eqref{E:eqforf}. After some algebra this gives the following solution of 
\eqref{E:Simp_Working_Hydro_Sleigh_Equations}:
 \begin{equation*}
 \begin{split}
\omega(t)&=2\sqrt{\frac{2h}{J}} \left ( \frac{K_2 \cos^2(\gamma t) +K_1 \sin (\gamma t)\cos (\gamma t)  }
{(1+K_2^2)\cos^2(\gamma t) + 2 K_1K_2\sin (\gamma t)\cos(\gamma t) +K_1^2\sin^2 (\gamma t)} \right ),  \\
v_1(t)&=\mbox{sign}(\alpha Z) \sqrt{\frac{2h}{M}} \left ( \frac{(K^2_2-1) \cos^2(\gamma t)+2K_1K_2 \sin (\gamma t)\cos (\gamma t) + K_1^2\sin^2(\gamma t)}
{(1+K_2^2)\cos^2(\gamma t) + 2 K_1K_2\sin (\gamma t)\cos(\gamma t) +K_1^2\sin^2 (\gamma t)} \right ).
\end{split}
\end{equation*}

We immediately obtain the following formula for the period $T$ of the solutions:
\begin{equation*}
\label{E:Period}
T=\frac{\pi}{\gamma}= \frac{\sqrt{2} \pi D}{\sqrt{E}\sqrt{h_0-h}}.
\end{equation*}
Notice that $T\rightarrow \infty$ as $h\rightarrow h_0$ from the left as expected.

We can also obtain a closed expression for the increment $\Delta \theta =\int_0^T \omega(t) \, dt$. For this matter
we note that a primitive of $\omega(t)$ is given by (see e.g.  \cite{GR}):
\begin{equation*}
\theta(t)=\frac{2}{\gamma}\sqrt{\frac{2h}{J}} \, \left ( \frac{\theta_1(t) +\theta_2(t) +\theta_3(t)}{4K_1^2K_2^2+(1+K_2^2-K_1^2)^2} \right ),
\end{equation*}
with
\begin{equation*}
\begin{split}
\theta_1(t)&=K_2(1+K_1^2+K_2^2)\gamma t, \\
\theta_2(t)&=\frac{K_1}{2}(K_1^2+K_2^2-1) \, \ln  ( (1+K_2^2)\cos^2(\gamma t) + 2 K_1K_2\sin (\gamma t)\cos(\gamma t) +K_1^2\sin^2 (\gamma t)  ), \\
\theta_3(t)&=-2K_1K_2\arctan (K_1\tan(\gamma t)+K_2).
\end{split}
\end{equation*}
Therefore, taking into account the change of branch in the expression for $\theta_3(t)$, we find that, under the above
assumption $L_2=0$,
\begin{equation}
\label{E:Deltaphi}
\Delta \theta = \theta(\pi/\gamma)-\theta(0)= \frac{2 \pi }{\gamma}\sqrt{\frac{2h}{J}} \, \left ( \frac{ K_2(1+K_1^2 +K_2^2)+2K_1K_2}{4K_1^2K_2^2+(1+K_2^2-K_1^2)^2} \right ).
\end{equation}
(In the general case $L_2\ne 0$ due to Proposition 4.1, $\Delta \theta$ will be a linear combination of \eqref{E:Deltaphi}) and of the integral $\int_0^T v_1(t) \, dt$.)

In particular notice that $\Delta \theta = O((h_0-h)^{-1/2})$ as $h\rightarrow h_0$ from the left. This
   implies that the behavior of the sleigh on the plane is extremely sensitive to the energy values that are slightly smaller than $h_0$. Namely, we have

\begin{proposition}
For any $\delta >0$ there exists a strictly increasing sequence of energy values
$\{ h^{(p)}_k \}_{k\in  \mathbb{N}}$ satisfying
$h_0-\delta<h^{(p)}_k<h_0$,
for which the motion of the contact point in the plane is periodic.
Similarly, there exists an increasing sequence of energy
 values $\{ h^{(u)}_k \}_{k\in  \mathbb{N}}$ (respectively, $\{ h^{(q)}_k \}_{k\in  \mathbb{N}}$) satisfying \\
$h_0-\delta<h^{(u)}_k<h_0$ (respectively,  $h_0-\delta<h^{(q)}_k<h_0$), such that the motion of the sleigh is 
unbounded (respectively, quasiperiodic).
\end{proposition}
The proof follows from the results  of Theorem \ref{T:Reconstruction_periodic} and the fact that 
$|\Delta \theta|\to \infty$ as $h\to h_0$ from the left.

\paragraph{The solution for $h=h_0$.} In this case the discriminant $\Delta = 0$, and \eqref{E:eqforf} 
has the solution \\ $f(t)=-\frac{1}{A}(\frac{1}{t}+B)$.
After some algebra this gives the following solution of \eqref{E:Simp_Working_Hydro_Sleigh_Equations}:
 \begin{equation*}
 \begin{split}
\omega(t)&=2\sqrt{\frac{2h_0}{J}}\,  \frac{K_3t +K_2 t^2 }
{(1+K_2^2)t^2 + 2 K_3K_2t+K_3^2}\,  ,  \\
v_1(t)&=\mbox{sign}(\alpha Z) \sqrt{\frac{2h_0}{M}} \, \frac{(K^2_2-1)t^2+2K_3K_2 t + K_3^2 }
{(1+K_2^2)t^2 + 2 K_3K_2t+K_3^2}  \, ,
\end{split}
\end{equation*}
where it is understood that $h=h_0$ in the expression for $K_2$.

To analyze the trajectory of the contact point of the sleigh, consider the
particular solution $\omega(t-t_0), \, v_1(t-t_0)$ with $t_0=-\frac{K_2K_3}{1+K_2^2}$.
In order to keep the notation as simple as possible, we will also denote this solution by
$\omega(t), \, v_1(t)$. We have
  \begin{equation*}
 \begin{split}
\omega(t)&=2\sqrt{\frac{2h_0}{J}} \left ( \frac{C_1^2K_2 t^2-K_3C_1C_2t-K_2K_3^2 }
{C_1(C_1^2 t^2 + K_3^2)} \right ),  \\
v_1(t)&=\mbox{sign}(\alpha Z) \sqrt{\frac{2h_0}{M}} \left ( \frac{C_1^2C_2t^2+4K_3K_2C_1 t - K_3^2C_2 }
{C_1(C_1^2 t^2 + K_3^2)}  \right ),
\end{split}
\end{equation*}
where $C_1=1+K_2^2$ and $C_2=K_2^2-1$.

Notice that
$$
\lim\limits_{t\to \pm\infty}  \frac{v_1}{\omega}  =
\frac{1}{2} \mbox{sign}(\alpha Z)\sqrt{\frac{J}{M}}  \frac{C_2}{K_2}  := R,
$$
hence the limit motions of the sleigh are along circles of the same radius $|R|$.
Moreover, the circles are traversed in the same direction as both $\omega(t)$ and
$v_1(t)$ have the same limits as $t\to \pm\infty$.

\begin{proposition} The motion of the contact point in the plane in the case $h=h_0$ is bounded.
\end{proposition}
\begin{proof}
It is sufficient to show that the distance between the centers of the limit circumferences  is finite.
The centers coincide with limit positions of the point $C$ of the sleigh with body coordinates $(0,R)$. Indeed, the components of the velocity of this point are given by
\begin{equation*}
\dot x_C= V_C (t) \cos (\theta (t)), \qquad \dot y_C= V_C (t) \sin (\theta (t)),
\end{equation*}
where
\begin{equation*}
V_C (t)=v_1(t)-R\omega(t) =\mbox{sign}(\alpha Z) \sqrt{\frac{2h_0}{M}}\frac{K_3(4K_2^2 +C_2^2)t}{K_2(C_1^2t^2 +K_3^2)},
\end{equation*}
that goes to zero as $t\to \pm \infty$. Therefore, the vector that connects the two centers has the components
\begin{equation*}
\Delta x_C= \int_{-\infty}^\infty V_C (t) \cos \theta (t) \, dt, \qquad
\Delta y_C= \int_{-\infty}^\infty V_C (t) \sin \theta (t) \, dt,
\end{equation*}
where $\theta(t)=\theta_1 (t) + \theta_2(t) +\theta_3(t)$ with
\begin{equation*}
\begin{split}
\theta_1 (t)&=2\sqrt{\frac{2h_0}{J}}\frac{K_2 t}{C_1}, \qquad
\theta_2 (t)=-\sqrt{\frac{2h_0}{J}}\frac{K_3 C_2}{C_1^2} \ln (C_1^2t^2+K_3^2)+\theta_0, \\
\theta_3(t)&=-4\sqrt{\frac{2h_0}{J}}\frac{K_3 K_2}{C_1^2} \arctan \left ( \frac{C_1t}{K_3} \right ),
\end{split}
\end{equation*}
$\theta_0$ being an integration constant. Notice that $\theta_1, \theta_3$ and $V_C$ are odd
functions of $t$ whereas $\theta_2$ is even. Therefore, using basic trigonometric identities, we get:
\begin{equation*}
\Delta x_C= -\int_{-\infty}^\infty  V_{C} \sin (\theta_1) \sin(\theta_2) \cos (\theta_3)  \, dt  - \int_{-\infty}^\infty  V_{C} \cos (\theta_1)  \sin (\theta_2)\sin (\theta_3) \, dt.
\end{equation*}
The above integrals can be shown to converge using integration by parts. For example,
using that $\sin (\theta_1) =\frac{d}{dt} ( -C_3\cos (\theta_1))$
for a certain constant $C_3$, the first integral rereads
\begin{equation*}
 \int_{-\infty}^\infty  V_{C} \sin (\theta_1) \sin(\theta_2) \cos (\theta_3)  \, dt  =
 -C_3\cos (\theta_1) V_{C}\sin(\theta_2) \cos (\theta_3) \big |_{-\infty}^\infty + C_3
  \int_{-\infty}^\infty   \cos (\theta_1)  k(t) \, dt,
\end{equation*}
where $k(t)=\frac{d}{dt} (V_{C} \sin(\theta_2) \cos (\theta_3))$.  It is seen that the two
 terms on the right are finite. Namely, the limits of the boundary terms are zero since $V_C\to 0$ as $t\to \pm \infty$.
 Next, the integral is absolutely convergent since $|k(t)|$ can be dominated
 by a function that decays as $1/t^2$ for large $|t|$.

 The proof that $\Delta y_C$ is finite is analogous.
\end{proof}

\paragraph{The solution for $h>h_0$.} In this case the discriminant $\Delta >0$ and  two solutions
of \eqref{E:eqforf} should be considered. These are $f_1(t)= \frac{1}{C}(-\gamma \tanh (\gamma t) -B)$ and
$f_2(t) =\frac{1}{C}(-\gamma \coth (\gamma t) -B)$, and correspond to the two heteroclinic orbits that, together with the equilibra, make up  the energy contour.

Starting out with $f_1(t)$ we find the solution to  \eqref{E:Simp_Working_Hydro_Sleigh_Equations}
\begin{equation}
\label{E:branch1}
 \begin{split}
\omega(t)&=2\sqrt{\frac{2h}{J}} \left ( \frac{K_2 \cosh^2(\gamma t) -K_1 \sinh (\gamma t)\cosh (\gamma t)  }
{(1+K_2^2)\cosh^2(\gamma t) - 2 K_1K_2\sinh (\gamma t)\cosh(\gamma t) +K_1^2\sinh^2 (\gamma t)} \right ),  \\
v_1(t)&=\mbox{sign}(\alpha Z) \sqrt{\frac{2h}{M}} \left ( \frac{(K^2_2-1) \cosh^2(\gamma t)-2K_1K_2 \sinh (\gamma t)\cosh (\gamma t) + K_1^2\sinh^2(\gamma t)}
{(1+K_2^2)\cosh^2(\gamma t) - 2 K_1K_2\sinh (\gamma t)\cosh(\gamma t) +K_1^2\sinh^2 (\gamma t)}  \right ),
\end{split}
\end{equation}
that is contained in the semi-plane $L_1\omega +Zv_1+\rho \alpha >0.$

Whereas considering  $f_2(t)$ one  finds the solution to  \eqref{E:Simp_Working_Hydro_Sleigh_Equations}
\begin{equation*}
 \begin{split}
\omega(t)&=2\sqrt{\frac{2h}{J}} \left ( \frac{K_2 \sinh^2(\gamma t) -K_1 \sinh (\gamma t)\cosh (\gamma t)  }
{(1+K_2^2)\sinh^2(\gamma t) - 2 K_1K_2\sinh (\gamma t)\cosh(\gamma t) +K_1^2\cosh^2 (\gamma t)} \right ),  \\
v_1(t)&=\mbox{sign}(\alpha Z) \sqrt{\frac{2h}{M}} \left ( \frac{(K^2_2-1) \sinh^2(\gamma t)-2K_1K_2 \sinh (\gamma t)\cosh (\gamma t) + K_1^2\cosh^2(\gamma t)}
{(1+K_2^2)\sinh^2(\gamma t) - 2 K_1K_2\sinh (\gamma t)\cosh(\gamma t) +K_1^2\cosh^2 (\gamma t)}  \right ),
\end{split}
\end{equation*}
that is contained in the semi-plane $L_1\omega +Zv_1+\rho \alpha <0.$

For the sequel, we only consider the branch of the solution corresponding to the expressions
for $\omega(t)$ and $v_1(t)$ given by \eqref{E:branch1}.

Integrating $\omega(t)$ and denoting $\kappa=2\sqrt{2h/J}$, we get
\begin{align*}
\theta (t) &= \int \omega \, dt = \theta_1+\theta_2+ \theta_3, \\
\theta_1 &= \kappa \frac{K_1+K_2}{ (K_1+K_2)^2+1 } \, t , \\
\theta_2 & = \kappa \frac {K_1}{2 \gamma  } \frac{ K_2^2-K_1^2-1 }{ [ (K_1+K_2)^2+1 ]\cdot[(K_1-K_2)^2+1 ] } \\
& \qquad \cdot \ln \left( [(K_1-K_2)^2+1]e^{4\gamma t}+[2(K_2^2-K_1^2+1)]e^{2\gamma t}+[(K_1+K_2)^2+1] \right), \\
\theta_3 &= \kappa \frac{2 K_1 K_2 }{\gamma   [ (K_1+K_2)^2+1 ]\cdot[(K_1-K_2)^2+1 ]}
\arctan \left( \frac{1 }{2 K_1} [ (K_1- K_2)^2+1] e^{2\gamma  t}+ \frac{1 }{2 K_1}[K_2^2-K_1^2+1] \right) .
\end{align*}

We will assume that  $\mbox{sign}(\alpha Z)=1$ in the expression for $v_1$.
Note that the radii of the limit circles are different:
$$
r_\pm = \lim_{t \to \pm \infty} \frac{v_1}{\omega}= \frac 12 \sqrt{\frac{J}{M}} \frac{ (K_2 \mp K_1)^2-1 }{K_2 \mp K_1}.
$$
Then, to evaluate the distance between their centers, introduce the ``floating radius''
$$
\rho(t) = \frac{r_+ e^{\gamma  t}+ r_- e^{-\gamma  t} }{ e^{\gamma  t}+ e^{-\gamma  t} }, \text{so that} \quad \rho(\pm \infty)= r_{\pm},
$$
and the floating point of the sleigh whose coordinates in the body fixed frame are $(0,\rho (t))$.
The limit positions of this point coincide with the centers of the limit circumferences.
The components of the velocity of the point are
\begin{equation*}
\dot x_\rho= V_\rho (t) \cos \theta (t), \qquad \dot y_\rho= V_\rho (t) \sin \theta (t),
\end{equation*}
where
$$
V_{\rho}(t) = v_1(t)- \rho (t) \omega(t)= \sqrt{\frac{2h}{M}} \frac {4 K_1^2} {K_2^2-K_1^2}
\cdot \frac{1}{ [(K_1-K_2)^2+1] e^{2 \gamma  t}+2 [ K_2^2-K_1^2-1] + [(K_1+K_2)^2+1] e^{-2\gamma  t}} ,
$$

Now introduce the new variable
$$
z = \frac{1 }{2 K_1} [ (K_1- K_2)^2+1 ] e^{2\gamma  t},
$$
and the constant
\begin{equation} \label{k_3}
k:=K_2^2-K_1^2+1=\frac{Eh_0 +JZ^2(2\sqrt{h h_1}+h_1)}{(\sqrt{h_1}+\sqrt{h})^2}.
\end{equation}
Then
\begin{gather}
e^{2\gamma  t} = \frac{ 2 K_1 z  }{(K_1- K_2)^2+1 }, \notag \\
dz = \frac{\gamma }{K_1}[(K_1-K_2)^2+1] e^{2\gamma  t} \, dt = 2\gamma z\, dt, \label{*} \\
z^2+\frac{k}{K_1}z+\frac{k^2}{4K_1^2}+1 =\frac{(K_1-K_2)^2+1}{4K_1^2}\left ( [(K_1-K_2)^2+1]e^{4\gamma t} +2(K_2^2-K_1^2+1)e^{2\gamma t}
+(K_1+K_2)^2+1 \right ).
\notag
\end{gather}
and, therefore,
\begin{align*}
\theta_1 & = \kappa \frac{K_1+K_2}{ (K_1+K_2)^2+1 } \frac{1}{2 \gamma } \ln \left (\frac{ 2 K_1 z  }{(K_1- K_2)^2+1 }\right ),  \\
\theta_2 & = \kappa \frac {K_1}{2 \gamma  } \frac{ K_2^2-K_1^2-1 }{ [ (K_1+K_2)^2+1 ]\cdot[(K_1-K_2)^2+1 ] } \ln \left (z^2+\frac{k}{K_1}z+\frac{k^2}{4K_1^2}+1 \right ) \quad
\text{(up to an additive constant)}, \\
\theta_3 & = \kappa \frac{2 K_1 K_2 }{\gamma   [ (K_1+K_2)^2+1 ]\cdot[(K_1-K_2)^2+1 ]} \arctan \left (z+\frac{k}{2K_1}\right ), \\
V_{\rho} & = \sqrt{\frac{2h}{M}} \, \frac{ [(K_1-K_2)^2+1]}{K_2^2 - K_1^2}\,  \frac{1}{z^2+\frac{k}{K_1}z+\frac{k^2}{4K_1^2}+1}\, e^{2\gamma  t} =
\sqrt{\frac{2h}{M}} \frac{1} {z^2+\frac{k}{K_1}z+\frac{k^2}{4K_1^2}+1} \frac{2 K_1z} {K_2^2 - K_1^2} \,  .
\end{align*}

The components of the vector of the distance between the centers of the limit circles are
$$
\Delta x_\rho = \int_{-\infty}^\infty V_{\rho} \cos\theta (t) \, dt , \quad
\Delta y_\rho = \int_{-\infty}^\infty V_{\rho} \sin\theta (t) \, dt
$$
Then, in view of \eqref{*},
\begin{align*}
\Delta x_\rho & = \sqrt{\frac{2h}{M}} \frac {2K_1} {K_2^2 - K_1^2}\, \int_{-\infty}^\infty
\frac{ z}{z^2+\frac{K_3}{K_1}z+\frac{K_3^2}{4K_1^2}+1} \cos\theta (t) \, dt \\
& = \sqrt{\frac{2h}{M}} \frac {K_1}{\gamma (K_2^2 - K_1^2) } \,\int_{0}^\infty \frac{\cos\theta}{z^2+\frac{K_3}{K_1}z+\frac{K_3^2}{4K_1^2}+1} dz ,
\end{align*}
and similarly for $\Delta y_\rho$.

Under the next substitution $z = \tan u$ the above integrals yield
\begin{gather} \label{complex_deltas}
\Delta x_C \pm i\Delta y_C
 = \frac {\sqrt{2h/M} \,K_1}{\gamma (K_2^2 - K_1^2) } a_0^{\pm i a_1} \int_{0}^{\pi/2} \frac{(\sin u)^{\pm ia_1}(\cos u)^{\mp i (a_1+2a_2)}\exp\left (\pm i a_3 \arctan \left ( \tan u + \frac{k}{2K_1} \right ) \right )}{
\left (1 + \frac{k}{2K_1}\sin (2u) + \frac{k^2}{4K_1^2}
 \cos ^2u \right )^{1 \mp ia_2}}\,  du, \\ \nonumber
 a_0 = \frac{2K_1}{(K_1-K_2)^2+1}, \quad a_1 = \frac{\kappa (K_1 +K_2)}{2\gamma((K_1+K_2)^2+1)},  \\ \nonumber
  a_2= \frac{\kappa K_1 (K_2^2-K_1^2-1)}{2((K_1+K_2)^2+1)((K_1-K_2)^2+1)},
 \quad a_3=\frac{2\kappa K_1 K_2}{\gamma ((K_1+K_2)^2+1)((K_1-K_2)^2+1)}.
\end{gather}
The latter integrals represent a generalization of the classical Beta-function (see, e.g., \cite{GR}). Their product gives the square of the distance between the centers.

\paragraph{Remark.} We could not calculate the above integrals explicitly. Note, however, that in the limit $h\to\infty$, when the energy prevails over the circulation effect, due to \eqref{E:def-coeff} and \eqref{k_3}, the integrals
\eqref{complex_deltas} reduce to
$$
\Delta x_C \pm i\Delta y_C
 = -\frac {2\sqrt{D} }{Z } c_0^{\mp i \zeta}
 \int_{0}^{\pi/2} (\sin u)^{\mp i\zeta}(\cos u)^{\mp i\zeta} \left(\cos (2c_1 \,u) \mp i \sin(2c_1 \,u)\right) \, du.
$$
$$
\zeta=\frac{DZ}{E}, \qquad c_1=\frac{\sqrt{D}ML_1}{E}, \qquad c_0=\frac{Z\sqrt{D}}{\sqrt{EM}+L_1M}.
$$
Note that the same reduction holds in case of zero circulation ($h_0=h_1=0$ in \eqref{k_3} implies $k=0$),
which was studied in \cite{hydro-sleigh}.  As was shown there, the product of the above integrals gives the
square of the distance between the centers in the following closed form
\begin{equation*}
\label{E:distance_sqd}
(\Delta x_C)^2 + (\Delta y_C)^2 = \frac{2\pi D}{Z^2} \left ( \frac{\zeta}{c_1^2+\zeta^2} \right ) \left ( \frac{\cosh ( \zeta \pi) - \cos (c_1 \pi )}{\sinh (\zeta \pi )} \right ).
\end{equation*}

\section*{Conclusions and further work}

We presented one of the first examples of nonholonomic hydrodynamical system, which   
is related to the design of underwater vehicles. The 
nonholonomic constraint can be interpreted as a first approximation model for a fin.  
From the mathematical point of view, our example is remarkable since both asymptotic and periodic 
dynamics coexist in the reduced phase space. It has been observed that
the value of the energy is a crucial parameter in the qualitative behavior
of the body on the plane. To our knowledge, this type of phenomenon is quite rare
in nonholonomic dynamics. 
(A similar behavior occurs in the classical Appel--Korteweg problem on the rolling disc \cite{App, Fed87}.)

For the future, we intend to study the motion of the hydrodynamic Chaplygin sleigh coupled to
{\em point vortices} in the fluid \cite{Ne01}.  The equations of motion for interacting point vortices 
and rigid bodies (without nonholonomic constraints) were recently derived in \cite{cylvortices, Bor_Mam_Ram} 
and since then there have been significant efforts towards discerning integrability and chaoticity \cite{Ra, RoAr2010} 
and towards uncovering the underlying geometry of these models \cite{Va_et_al}.  
We plan to couple the nonholonomic Chaplygin sleigh with one or several point vortices in the flow, 
taking these models as our next starting point. 
\medskip

\paragraph{Acknowledgments} 

{\small
We thank the GMC (Geometry, Mechanics and Control Network, project
MTM2009-08166-E, Spain) for facilitating our collaboration during the events that it organizes. 

Yu.N.F's contribution was partially supported by the project MICINN MTM2009-06973. 
  
We are also thankful to
Hassan Aref, Francesco Fass\`o, Andrea Giacobbe, and Paul Newton for useful and interesting discussions,
and to Maria Przybylska who helped us to integrate the reduced system \eqref{E:Simp_Working_Hydro_Sleigh_Equations}.}

{\small
JV is a postdoc at the Department of Mathematics of UC San Diego, partially supported by NSF CAREER award 
DMS-1010687  and NSF FRG grant DMS-1065972, and by the {\sc irses} project {\sc
geomech} (nr.\ 246981) within the 7th European Community Framework Programme, 
and is on leave from a Postdoctoral Fellowship of the Research Foundation--Flanders (FWO-Vlaanderen). 
LGN acknowledges the hospitality of the  Department de Matem\'atica Aplicada I, 
and IV, at UPC Barcelona for his recent stay there. }

\end{document}